\documentclass[11pt]{article}

\usepackage{authordate1-4}
\usepackage{FRdina4}
\usepackage{amssymb}  
\usepackage{MyMacros}
\usepackage{color}
\usepackage{url}

\begin{document}
\pagestyle{headings}

\title{\textbf{Information and Set Algebras:\\Interpretation and Uniqueness of Conditional Independence}}

\author{Juerg Kohlas \\
\small Department of Informatics DIUF \\ 
\small University of Fribourg \\ 
\small CH -- 1700 Fribourg (Switzerland) \\ 
\small E-mail: \texttt{juerg.kohlas@unifr.ch} \\
\small \texttt{http://diuf.unifr.ch/drupal/tns/juerg\_kohlas
}
}
\date{\today}

\maketitle


\begin{abstract}
A new seemingly weak axiomatic formulation of information algebras is given. It is shown how such information algebras can be embedded into set (information) algebras. In set algebras there is a natural relation of conditional independence between partitions. Via the embedding of information algebras this relation carries over to information algebras. The new axiomatic formulation is thereby shown to be equivalent to the one given in \cite{kohlas17}. In this way the abstract concept of conditional independence in information algebras gets a concrete interpretation in terms of set theoretical relations.
\end{abstract}

\tableofcontents


\section{Introduction and Overview}

In \cite{kohlas03,kohlasschmid14,kohlasschmid16,kohlas17} various axiomatic definitions of information algebras are given. Although these formulations are not fully equivalent, they all model the idea that information comes in pieces, can be aggregated or combined, that pieces of information refer to questions or domains, and that the part of a piece of information related to a given question can be extracted. The operations of combination and extraction are subject to some axiomatic requirements and define thus algebraic structures called information algebras. The axiomatic formulation given in \cite{kohlas17} is so far the most general one. It is based on a relation of \textit{conditional independence} between domains. The concept of conditional independence is fundamental for any formalism modelling information, as is known from probability theory and the theory of relational databases for example, and in many more systems, see for instance \cite{studeny93,shenoy94b,studeny95}. In \cite{dawid01} a fundamental mathematical structure called separoid, capturing the essence of conditional independence, is discussed. The axiomatic structure proposed in \cite{kohlas17} is based on a slightly more general structure, called \textit{quasi-separoid}. This is a purely abstract formulation of conditional independence. The purpose of this paper is to give a more concrete interpretation of this concept, based on set-theoretic concepts. This is achieved by representing abstract information algebras by \textit{set algebras}, based on set-theoretic operations. 

The paper starts with yet another axiomatic system for information algebras (Section \ref{sec:FomFreeAlg}), which is seemingly much weaker than the system proposed in \cite{kohlas17}. A main result of the paper is to show that in fact it is equivalent to the one in \cite{kohlas17}. To show this, the concept of set (information) algebras is introduced in Section Ê\ref{sec:SetAlgebra}. It is then shown that any information algebra in the sense of Section \ref{sec:FomFreeAlg} can be embedded into a set algebra, that is, is isomorphic to some set algebra. This is discussed using the concept of order-generating sets (or meet-dense sets) (Section \ref{sec:InfAndSetAlg}). It turns out that, depending on structural properties of the information algebra, different order-generating sets and hence different embeddings exist (Section \ref{sec:Exmpls}). The essential point is then that between partitions of a set a natural conditional independence relation can be defined. This relation forms a quasi-separoid. On the base of this relation two computationally import properties of set algebras, the combination and extraction properties, can be derived. Via the set algebra representations of an information algebra, these q-separoids of partitions induce a conditional independence relation among domains of the information algebra and the combination and extraction properties of the set algebra are via the embedding inherited in the information algebra too. In \cite{kohlas17}  these properties were postulated as axioms, whereas we show here that they can be derived from weaker postulates. Even if there are several different embeddings of an information algebra into different set algebras, it turns out that the conditional independence relation induced in the information algebra via the sets algebras is unique, does not depend on the particular set algebra used to induce it. This is discussed in Section \ref{sec:CondIndep} and constitutes the main result of this paper. As a further consequence, the present axiomatic definition of an information algebra covers all previous formulations \cite{kohlas03,kohlasschmid14,kohlasschmid16,kohlas17}, which turn out to be special cases.

Representations of information algebras by set algebras have already been discussed in \cite{kohlas03,kohlasschmid16} in the context of their respective axiomatic systems. The set representations presented here generalize the representations of the papers cited. Since it is shown that the system of information algebras proposed in \cite{kohlas17} can be derived from the present axiomatic system, all results in this reference remain valid.


\section{Information Algebras} \label{}

\subsection{Domain-free Algebras} \label{sec:FomFreeAlg}

In an algebraic view of information we consider first that information comes in pieces, which can be combined or aggregated to new pieces of information. Later on, we shall argue that pieces of information refer to certain questions, represent (at least partial) answers to them and foremost that from a piece of information the part relating to some given questions can be extracted. 

So, in a first step let $\Psi$ be a set of elements representing pieces of information. Pieces of information can be combined, an operation in $\Psi$ which is represented by a binary operation 
\begin{eqnarray*}
\cdot : \Psi \times \Psi \rightarrow \Psi, \quad (\phi,\psi) \mapsto \phi  \cdot \psi.
\end{eqnarray*}
The product $\phi \cdot \psi$ is thought to represent the piece of information obtained by combining $\phi$ and $\psi$. This operation is assumed to be associative, commutative and idempotent. Further, it is assumed that there is a unit element $1 \in \Psi$ such that $\psi \cdot 1 = 1 \cdot \psi = \psi$ for all elements $\psi$ of $\Psi$. The unit element represents vacuous information, combining it with any other piece of information changes nothing. In addition, a null element $0$ is also assumed in  $\Psi$, such that $\psi \cdot 0 = 0 \cdot \psi = 0$. The element $0$ represents contradictory information, which by combining with any other information destructs it. So, the signature $(\Psi;\cdot,1,0)$ represents a commutative. idempotent semigroup with unit and null element. 

The idempotency of the semigroup permits to introduce an order into $\Psi$. In fact, this can be done in two ways. We choose to define $\phi \leq \psi$ if $\phi \cdot \psi = \psi$. The idea is  that $\phi$ is less informative than $\psi$, if combining with $\psi$ gives nothing new; $\psi$ is then more informative than $\phi$. It can easily be verified that $\leq$ is a partial order in $\Psi$, $1$ is the smallest and $0$ the largest element in this order. We call this the \textit{information order} in $\Psi$. It can also readily be verified that combination of two elements results in the \textit{supremum} of the two element in this order,
\begin{eqnarray*}
\phi \cdot \psi = \sup\{\phi,\psi\} = \phi \vee \psi.
\end{eqnarray*}
So, the semigroup $(\Psi;\cdot,1,0)$ can also be seen as a bounded join-semilattice $(\Psi;\leq)$. Depending on the case we shall stress the combination or the order aspects of $\Psi$.

Next we consider the second aspect of pieces information, namely that they refer to questions. At this place we do not try to describe the internal structure of questions, we rather think of questions as represented by some abstract domains, describing or representing somehow the possible answers to the questions associated; in the simplest case for instance by listing the possible answers. Let then $D$ be a set whose elements are thought to represent questions or domains. Its generic elements will denoted by lowercase letters like $x$, $y$, $z, \ldots$. We assume however that domains or questions can be compared with respect to their granularity or fineness. Therefore, we require $(D;\leq)$ to be a partial order, where $x \leq y$ means that $y$ is finer than $x$, that is, answers to $y$ will be more informative than answers to $x$. Moreover, if $x$ and $y$ are two elements of $D$, we want to able to consider the combined question represented by $x$ and $y$. This combined question is surely finer than both $x$ and $y$; it is in fact the coarsest questions finer than $x$ and $y$, that is the supremum of $x$ and $y$ or the join $x \vee y$ with respect to the order in $D$. So, we assume $(D;\leq)$ to be a join-semilattice. 

Now, it must be possible to extract from any piece of information $\psi$ in $\Psi$ the part relating to a domain or question $x$. This is achieved by \textit{extraction} maps $\epsilon_x : \Psi \rightarrow \Psi$, where $\epsilon_x(\psi)$ represents the part of $\psi$ referring to domain $x$. Let $\mathcal{E}$ be the set of extraction maps $\{\epsilon_x:x \in D\}$. We require for each extraction map $\epsilon_x$ to satisfy the following conditions:
\begin{description}
\item[E1] $\epsilon_x(0) = 0$,
\item[E2] $\psi \cdot \epsilon_x(\psi) = \psi$ for all $\psi \in \Psi$,
\item[E3] $\epsilon_x(\epsilon_x(\phi) \cdot \psi) = \epsilon_x(\phi) \cdot \epsilon_x(\psi)$.
\end{description}
E1 says that contradiction cannot be eliminated by extraction. E2 states that information extracted from a piece of information $\psi$ is contained in $\psi$. E3 requires that the part relative to a domain $x$ of a combination of a piece of information relating to $x$ with any other piece, can be obtained by combining the piece relating to $x$ with the part relating to $x$ of the second piece of information. It is a crucial condition as we shall see. 

We may restate this conditions also in order-theoretic terms:
\begin{description}
\item[E1] $\epsilon_x(0) = 0$,
\item[E2] $\epsilon_x(\psi) \leq \psi$ for all $\psi \in \Psi$,
\item[E3] $\epsilon_x(\epsilon_x(\phi) \vee \psi) = \epsilon_x(\phi) \vee \epsilon_x(\psi)$.
\end{description}
An operator satisfying these three conditions is called an \textit{existential quantifier} in algebraic logic, although in the relevant literature the opposite order rather than our information order is used to define existential quantifiers.

A domain $x \in D$ is called a support of $\psi \in \Psi$ if $\epsilon_x(\psi) = \psi$. We add two additional requirements for extraction operators:
\begin{description}
\item[E4] $\forall \psi \in \Psi$, there is a $x \in D$ so that $\epsilon_x(\psi) = \psi$,
\item[E5] If $\epsilon_x(\psi) = \psi$ and $x \leq y$, then $\epsilon_y(\psi) = \psi$.
\end{description}
Every $\psi \in \Psi$ has a support $x$ or relates fully to some domain $x$. This means that any piece of information refers at least to one of the questions or domains in $D$. If $\psi$ has a support $x$ and $x \leq y$, then $\psi$ has also support $y$; if $\psi$ refers to a domain $x$, then it refers also to any finer domain $y$.

A system $(\Psi;\mathcal{E},\cdot,1,0)$, where $(\Psi;\cdot,1,0)$ is a idempotent, commutative semigroup and $\mathcal{E}$ a family of operators $\epsilon_x$ for $x \in D$, satisfying conditons E1 to E5 relative to a join-semilattice $(D;\leq)$ is called an \textit{information algebra} or more precisely, a \textit{domain-free} information algebra. There is also a related version, called a \textit{labeled} information algebra. In \cite{kohlas03} and \cite{kohlasschmid16} different axiomatic formulations of information algebras are given, which we shall show to be special cases of the present one. Also in \cite{kohlas17} still another axiomatic is given, which turns below out to be essentially equivalent to the one above. In \cite{kohlas03} and \cite{kohlas17} the labeled versions of information algebras are presented; here we shall not discuss the labeled version.

Here follow a few elementary properties of support and extraction:

\begin{lemma} \label{le:SuppProp}
If $(\Psi;\mathcal{E},\cdot,1,0)$ is an information algebra, then the following holds for $x,y \in D$ and $\phi,\psi \in \Psi$:
\begin{enumerate}
\item $\epsilon_x(1) = 1$,
\item $\phi \leq \psi$ implies $\epsilon_x(\phi) \leq \epsilon_x(\psi)$,
\item $x$ is a support of $\epsilon_{x}(\phi)$, $\epsilon_{x}(\epsilon_{x}(\phi)) = \epsilon_{x}(\phi)$,
\item if $x \leq y$, then $\epsilon_x(\psi) \leq \epsilon_y(\psi)$ for all $\psi \in \Psi$,
\item if $x \leq y$, then $\epsilon_x(\epsilon_y(\psi)) = \epsilon_x(\psi)$,
\item if $x$ is a support of both $\phi$ and $\psi$, then it is also a support of $\phi \cdot \psi$, $\epsilon_{x}(\phi \cdot \psi) = \phi \cdot \psi$,
\item if $x$ is a support of $\phi$ and $y$ of $\psi$, then $x \vee y$ is a support of $\phi \cdot \psi$, $\epsilon_{x \vee y}(\phi \cdot \psi) = \phi \cdot \psi$.
\end{enumerate}
\end{lemma}

\begin{proof}
1.) By E2 we have $1 \cdot \epsilon_x(1) = 1$, hence $\epsilon_x(1) \leq 1$, but we have also $1 \leq \epsilon_x(1)$, since the unit is the smallest element in $(\Psi,\leq)$, therefore $\epsilon_x(1) = 1$.

2.) $\phi \leq \psi$ means $\phi \cdot \psi = \psi$. Hence we have by E3 und E2
\begin{eqnarray*}
\epsilon_x(\phi) \cdot \epsilon_x(\psi) = \epsilon_x(\epsilon_x(\phi) \cdot \psi) = \epsilon_x(\epsilon_x(\phi) \cdot \phi \cdot \psi) = \epsilon_x(\phi \cdot \psi) = \epsilon_x(\psi).
\end{eqnarray*}
This shows that $\epsilon_x(\phi) \leq \epsilon_x(\psi)$.

3.) This follows from E3 and item 1 in the following way:
\begin{eqnarray*}
\epsilon_x(\epsilon_x(\psi)) = \epsilon_x(\epsilon_x(\psi) \cdot 1) = \epsilon_x(\psi) \cdot \epsilon_x(1) = \epsilon_x(\psi) \cdot 1 = \epsilon_x(\psi).
\end{eqnarray*}

4.) Since $x$ is a support of $\epsilon_x(\psi)$ and $\epsilon_x(\psi) \leq \psi$, we have by E5 and item 2
\begin{eqnarray*}
\epsilon_x(\psi) \cdot \epsilon_y(\psi) = \epsilon_y(\epsilon_x(\psi)) \cdot \epsilon_y(\psi) = \epsilon_y(\psi).
\end{eqnarray*}

5.) Since $\epsilon_x(\psi) \leq \epsilon_y(\psi)$ (item 4) we have by E3 $\epsilon_x(\epsilon_y(\psi)) = \epsilon_x(\epsilon_x(\psi) \cdot \epsilon_y(\psi)) = \epsilon_x(\psi) \cdot \epsilon_x(\epsilon_y(\psi))$. But $\epsilon_y(\psi) \leq \psi$, hence $\epsilon_x(\epsilon_x(\psi)) \leq \epsilon_x(\psi)$ and therefore we conclude that $\epsilon_x(\epsilon_y(\psi)) = \epsilon_x(\psi)$.

6.) Assuming $x$ is a support of $\phi$ and $\psi$, using E3 we obtain
\begin{eqnarray*}
\epsilon_x(\phi \cdot \psi) = \epsilon_x(\epsilon_x(\phi) \cdot \psi) = \epsilon_x(\phi) \cdot \epsilon_x(\psi) = \phi \cdot \psi.
\end{eqnarray*}
So $x$ is also a support for $\phi \cdot \psi$.

7.) By E5 $x \vee y$ is a support both of $\phi$ and $\psi$. Then the claim follows from item 4 above.
\end{proof}

We use these results in the sequel without explicit reference. Examples for information algebras may be found in the references \cite{kohlas03,kohlasschmid16,kohlas17} although with respect to less general axiomatics. Further examples will be presented below in due course.

\subsection{Homomorphisms and Subalgebras}

Consider two information algebras $(\Psi_1;\mathcal{E}_1,\cdot_1,1_1,0_1)$ and $(\Psi_2;\mathcal{E}_2,\cdot_2,1_2,0_2)$, where the first is based on a join-semilattice $(D_1;\leq_1)$ whereas the second one on the join-semilattice $(D_2;\leq_2)$. Here we define what we understand by a homomorphism of the first algebra into the second one.

\begin{definition}
\textit{Homomorphism:} A pair of maps $(f,g)$,
\begin{eqnarray*}
f : \Psi_1 \rightarrow \Psi_2, \quad g : D_1 \rightarrow D_2,
\end{eqnarray*}
is called a homomorphism between the two information algebras $(\Psi_1;E_1,\cdot_1,1_1,0_1)$ and $(\Psi_2;E_2,\cdot_2,1_2,0_2)$, if
\begin{enumerate}
\item $g$ is a join-homomorphism, i.e.for $x,y \in D_1$, $x \leq_1 y$ implies $g(x) \leq_2 g(y)$ in $D_2$, and $g(x \vee_1 y) = g(x) \vee_2 g(y)$,
\item $f$ preserves combination (or join), null and unit i.e. 
\begin{enumerate}
\item for $\phi,\psi \in \Psi_1$, $f(\phi \cdot_1 \psi) = f(\phi) \cdot_2 f(\psi)$,
\item $f(0_1) = 0_2$ and $f(1_1) = 1_1$,
\end{enumerate}
\item for $\psi \in \Psi_1$ and $x \in D_1$,
\begin{eqnarray*}
f(\epsilon_x(\psi)) = \epsilon_{g(x)}(f(\psi)).
\end{eqnarray*}
\end{enumerate}
If $g$ is a bijection (one-to-one and onto) and $f$ an injection (one-to-one) then $(f,g)$ is called an embedding and if $f$ is also a bijection, then $(f,g)$ is called an information-algebra isomorphism and the two algebras are called isomorphic.
\end{definition}

Next, we turn to subalgebras.

\begin{definition}
\textit{Subalgebra:} If $\Psi_1$ and $D_1$ are subsets of $\Psi_2$ and $D_2$ such that the inclusion maps define a homomorphism, then $(\Psi_1;\mathcal{E}_1,\cdot_1,1_1,0_1)$ is a subalgebra of  $(\Psi_2;\mathcal{E}_2,\cdot_2,1_2,0_2)$. That is, $\Psi_1$ is closed in $\Psi_2$ under formation of combination and extractions from $D_1$, and contains the null and unit elements $0$ and $1$, whereas $D_1$ is closed in $D_2$ under formation of joins
\end{definition}

A subalgebra of an information algebra is again an information algebra. We remark that in an information algebra $(\Psi;\mathcal{E},\cdot,1,0)$ the sets $\epsilon_x(\Psi) = \{\psi \in \Psi:\epsilon_x(\psi) = \psi\}$ of all elements with support $x$ and $D_x = \{y \in D:y \leq x\}$ define a subalgebra of $(\Psi;\mathcal{E},\cdot,1,0)$. Clearly $D_x$ is closed in $D$ under joins and $\epsilon_x(\Psi)$ is closed in $\Psi$ under combination, since by Lemma \ref{le:SuppProp}, we have $\epsilon_x(\phi \cdot \psi) = \phi \cdot \psi$ if $\phi$ and $\psi$ have support $x$. Further, if $\epsilon_x(\psi) = \psi$ and $y \leq x$, again by Lemma \ref{le:SuppProp}, since $\epsilon_y(\psi)$ has support $y$, we have by E5 $\epsilon_x(\epsilon_y(\psi)) = \epsilon_y(\psi)$ and $\epsilon_x(\Psi)$ is closed under extraction operators for $y \in D_x$. So, if $\mathcal{E}_x$ is the set of extraction operators $\epsilon_y$ with $y \leq x$, then $(\epsilon_x(\Psi);\mathcal{E}_x,\cdot,1,0)$ is still an information algebra, a subalgebra of $(\Psi;\mathcal{E},\cdot,1,0)$.

\subsection{Ideal Completion} \label{subsec:IdCompl}

In an information algebra $(\Psi;\mathcal{E};\cdot,1,0)$, a \textit{consistent} set of pieces of information $I$ is a nonempty subset $I$ of $\Psi$ such that (i) with any element $\phi \in I$ also all elements $\psi \leq \phi$ implied by $\phi$ or contained in $\phi$ belong to $I$, and (ii) with any two elements $\phi,\psi \in I$ also their combination $\phi \cdot \psi$ belongs to $I$. Such sets are just \textit{ideals} in the context of the join-semilattice $(\Psi;\leq)$. Ideals not equal to $\Psi$ are called proper. Consistent sets (also called \textit{theories}) may also be thought of as pieces of information. In fact, we may define among them operations of combination and extraction as follows.

Let $I_{\Psi}$ denote the family of all ideals contained in $\Psi$. We define the following two operations for ideals $I_{1},I_{2},I\in I_{\Psi}$ and $x \in D$:

\begin{enumerate}
\item \textit{Combination:} $I_{1} \cdot I_{2} = \{\phi \in \Psi:\phi \leq \phi_{1} \cdot \phi_{2} \textrm{ for some}\ \phi_{1} \in I_{1},\phi_{2} \in I_{2}\}$,
\item \textit{Extraction:} $\bar{\epsilon}_x(I) = \{\phi \in \Psi:\phi \leq \epsilon_x(\psi) \textrm{ for some}\ \psi \in I\}$.
\end{enumerate}

It turns out that the system $(I_{\Psi};\bar{\mathcal{E}},\cdot,\{1\},\Psi)$ with $\bar{\mathcal{E}} = \{\bar{\epsilon}_x:x \in D\}$ is an information algebra \cite{kohlas03,kohlasschmid14}, called the \textit{ ideal completion} of $(\Psi;\mathcal{E},\cdot,1,0)$. Moreover, the original algebra $(\Psi;\mathcal{E},\cdot,1,0)$ may be embedded into its ideal completion by the map $\psi \mapsto \downarrow\!\psi$, where the \textit{down-set} $\downarrow\!\psi = \{\phi: \phi\leq \psi\}$ is the principal ideal generated by $\psi$. Ideal completions will play an important role for Boolean information algebras (Section \ref{subsec:GenBooleInfAlg}) and distributive lattice information algebras (Section \ref{subsec:GenDistrLattInfAlg}). It is well-known that $I_{\Psi}$, ordered by set inclusion, is a complete lattice.

For later reference, we need the following result:

\begin{lemma} \label{le:ExtrOfideals}
Let $(\Psi;\mathcal{E},\cdot,1,0)$  be an information algebra and $(I_{\Psi};\bar{\mathcal{E}},\cdot,\{1\},\Psi)$ its ideal completion. Then for $x \in D$,  $\bar{\epsilon}_x(I) = \bar{\epsilon}_x(J)$ iff $I \cap \epsilon_x(\Psi) = J \cap \epsilon_x(\Psi)$.
\end{lemma}

\begin{proof}
Assume $\bar{\epsilon}_x(I) = \bar{\epsilon}_x(J)$ and consider $\psi \in I \cap \epsilon_x(\Psi)$. Then we have also $\psi \in \bar{\epsilon}_x(I)$, hence $\psi \in \bar{\epsilon}_x(J)$. But since $\psi = \epsilon_x(\psi)$ and $\bar{\epsilon}_x(J) \subseteq J$,  we have also $\psi \in J \cap \epsilon_x(\Psi)$. By symmetry this implies $I \cap \epsilon_x(\Psi) = J \cap \epsilon_x(\Psi)$.

Conversely, assume $I \cap \epsilon_x(\Psi) = J \cap \epsilon_x(\Psi)$ and consider $\phi \in \bar{\epsilon}_x(I)$. Then there is a $\psi \in I$ such that $\phi \leq \epsilon_x(\psi)$. But then $\epsilon_x(\psi) \in I \cap \epsilon_x(\Psi)$. Now we have $\phi \leq \epsilon_x(\psi) \in J \cap \epsilon_x(\Psi)$ and this implies $\phi \in \bar{\epsilon}_x(J)$. By symmetry it follows that $\bar{\epsilon}_x(I) = \bar{\epsilon}_x(J)$, concluding the proof.
\end{proof}

In the the next section an important and basic class of information algebras is introduced, later, in Section \ref{subsc:ExplInfAlg} some further examples are presented.

\section{Set Algebras} \label{sec:SetAlgebra}

So far the set $\Psi$ of pieces of information as well as the set $D$ of domains have been arbitrary abstract sets, subject only to the axioms specified for combination and extraction. We will now define a special type of an information algebras, called set algebras, whose information elements are subsets of some universe and the information operations are described by set-theoretical constructs. It will then be discussed in this paper to what extend general, abstract information algebras can be represented by or identified to such special set algebras.

We consider a base set $U$ ($U \not= \emptyset$) , the universe, which can be visualized as a set of possible worlds. The power set of $U$ will be denoted by $2^U$. Domains $x$, representing questions, will be modeled by equivalence relations $\equiv_x$ on $U$. The idea is that for $u$ and $u'$ in $U$ we have $u \equiv_x u'$ iff question $x$ has the same answer in worlds $u$ and $u'$. Equivalent relations $\equiv_x$ induce partitions $P_x$ of the base set $U$ whose blocks are the equivalence classes $[u]_x$ of the equivalence relation. A question $x$ will be considered to be finer than a question $y$, iff $u \equiv_x u'$ implies $u \equiv_y u'$, or equivalently, iff every block of $P_y$ is contained in a (unique) block of $P_x$. We denote this situation by $x \leq y$, read as ``$x$ is coarser than $y$'' or ``$y$ is finder than $x$''. It is obvious that $\leq$ defines a partial order in the family of all equivalence relations on $U$ respective all partitions of $U$, $Part(U)$. The order $(Part(U);\leq)$ has the partition $\{\{u\}:u \in U\}$, where all blocks consist of a single element as top element, as finest partition, and $\{U\}$ the partition consisting of the single block $U$ as coarsest element.

This order between partitions is motivated by information-theoretic considerations. In lattice theory usually the opposite order is considered. But in both cases, its is well known that the order $(Part(U);\leq)$ is a lattice, where the join $P_1 \vee P_2$ of any two partitions $P_1$ and $P_2$ in our information order is the partition consisting of all blocks of the form $B_1 \cap B_2 \not= \emptyset$ where $B_1 \in P_1$ and $B_2 \in P_2$. The meet operation is a bit more involved, its discussion is postponed to a later part, since we do not need it for the moment. 

Now, consider a join-sublattice $(\mathcal{D};\leq)$ of $(Part(U);\leq)$. To any partition of the universe $U$ a saturation operator $\sigma_P$ defined by
\begin{eqnarray*}
\sigma_P(X) = \{u \in U:\exists u' \in U \textrm{ such that}\ u \equiv_P u'\},
\end{eqnarray*}
where $u \equiv_P u'$ if $u$ and $u'$ belong to the same block of $P$. Let the $\mathcal{E}$ be the set of all saturation operators $\sigma_P$ for $P \in \mathcal{D}$. Now, any saturation operator is an existential quantifier relative to any partition in $Part(U)$, which follows from the following lemma.

\begin{lemma}   \label{saturation operators}
Let $\sigma_P$, $P\in Part(U)$, be a saturation operator on $U$. Then for all $X,Y\subseteq U$
\begin{enumerate}
\item  $\sigma_P(\emptyset) =  \emptyset$,
\item $X\subseteq\sigma_P(X)$,
\item$X\subseteq Y$ implies $\sigma_P(X)\subseteq\sigma_P(Y)$,
\item $\sigma_P(\sigma_P(X)\cap Y) = \sigma_P(X)\cap\sigma_P(Y)$.
\end{enumerate}
\end{lemma}

\begin{proof}
For 1. we have  $\sigma_P(\emptyset) = \bigcup\{B\in P: B\cap\emptyset\neq\emptyset\} = \emptyset$.

Items 2. and 3. are obvious.

For 4., observe that $\sigma_P(X)\cap Y\subseteq \sigma_P(X)\cap\sigma_P(Y)$,  so $\sigma_P(\sigma_P(X)\cap Y) \subseteq \sigma_P(X)\cap\sigma_P(Y)$ by 3. Now $\sigma_P(X)\cap\sigma_P(Y)$ is the union of all $B\in P$ satisying $B\cap X\neq\emptyset\neq B\cap Y$. Obviously, for each such $B$ we have $B\cap\sigma_P(X) = B$, so $B\cap\sigma_P(X) \cap Y\neq \emptyset$ and $B$ participates in the union of all $B'\in P$ forming  $\sigma_P(\sigma_P(X)\cap Y)$.
\end{proof}

It follows from this lemma that the elements of $\mathcal{E}$ satisfy requirements E1 to E3 of extraction operators. Now consider a family $\mathcal{S}$ of subsets of $U$ which are saturated with respect to a partition in $\mathcal{D}$, that is are a union of blocks of some partition $P \in \mathcal{D}$. That is, any set $X \in \mathcal{S}$ has a a partition $P$ in $\mathcal{D}$ as support, $\sigma_P(X) = X$. So requirement E4 for extraction operators is satisfied. Further if $P' \leq P$ and $P'$ is a support for $X$, then so is $P$; therefore E5 holds too. Finally $\mathcal{S}$ is obviously closed under intersection, contains the empty set and the universe $U$ as null and unit elements. So we conclude that $(\mathcal{S};\mathcal{E},\cap,U,\emptyset)$ is an information algebra. It is called a \textit{set algebra}, because its elements are subsets and combination and extraction are set-theoretical operations. The signature $(2^U;\mathcal{E},\cap,U,\emptyset)$ is also an information algebra \textit{except} that E4 does not hold in general, there may be subsets of $U$ which are saturated for no saturation operator in $\mathcal{E}$. Nevertheless, we call these weaker systems also set algebras. The point is, that if any information algebra is embedded into such a set algebra $(2^U;\mathcal{E},\cap,U,\emptyset)$, its image is a subalgebra which satisfies E4.

We now define a relation of (conditional) independence between partitions. For a finite set of partitions $P_1,\ldots,P_n$, $n \geq 2$ define
\begin{eqnarray*}
R(P_1,\ldots,P_n) = \{(B_1,\ldots,B_n):B_i \in P_i,\cap_{i=1}^n B_i \not= \emptyset\}.
\end{eqnarray*}
So, $R$ contains the tuples of mutually compatible blocks, representing compatible answers to the $n$ questions modelled by the partitions $P_1,\ldots,P_n$. We call the partitions \textit{independent}, if $R(P_1,\ldots,P_n)$ is the Cartesian product of $P_1,\ldots,P_n$,
\begin{eqnarray*}
R(P_1,\ldots,P_n) = P_1 \times \cdots \times P_n.
\end{eqnarray*}
This means that if an answer to a question $P_i$ is known to be in some block $B_i$, this does not constrain the answers to the other questions, or in other words, the answer to question $P_i$ contains no information relative to the other questions $P_1,\ldots,P_n$. Further, if $B$ is a block of a partition $P$ (contained or not in $P_1,\ldots,P_n$), then define for $n \geq 1$,
\begin{eqnarray*}
R_B(P_1,\ldots,P_n) = \{(B_1,\ldots,B_n):B_i \in P_i,\cap_{i=1}^n B_i \cap B \not= \emptyset\}.
\end{eqnarray*}
This represents the tuples of blocks of $P_1,\ldots,P_n$ compatible among themselves and with block $B \in P$. We call $P_1,\ldots,P_n$ \textit{conditionally indpendent} given $P$, if
\begin{eqnarray*}
R_B(P_1,\ldots,P_n) = R_B(P_1) \times \cdots \times R_B(P_n).
\end{eqnarray*}
So, knowing an answer to $P_i$, compatible with $B \in P$, gives no information on the answers to the other questions, except that they must each be compatible with $B$. Note that if this relation holds, then $B_i \cap B \not= \emptyset$ for all $i=1,\ldots,n$, imply that $B_1 \cap \ldots \cap B_n \cap B \not= \emptyset$. In this case we write $\bot\{P_1,\ldots,P_n\} \vert P$, or, for $n =2$, $P_1 \bot P_2 \vert P$. We may also say that  $P_1 \bot P_2 \vert P$, if $u \equiv_P u'$, implies that there is an element $v \in U$ such that $u \equiv_{P_1 \vee P} v$ and $u' \equiv_{P_2 \vee P} v$.

The three-place relation $P_1 \bot P_2 \vert P$ among partitions has the following properties:

\begin{theorem} \label{th:QSepOfPart}
\
\begin{description}
\item[C1] $P_1 \bot P_2 \vert P_2$,
\item[C2] $P_1 \bot P_2 \vert P$ implies $P_2 \bot P_1 \vert P$,
\item[C3] $P_1 \bot P_2 \vert P$ and $Q \leq P_2$ implies $P_1 \bot Q \vert P$,
\item[C4] $P_1 \bot P_2 \vert P$ implies $P_1 \bot P_2 \vee P \vert P$
\end{description}
\end{theorem}

\begin{proof}
C1 and C2 are obvious. To prove C3 assume $P_1 \bot P_2 \vert P$ and $Q \leq P$. Then $u \equiv_P u'$ implies the existence of an element $v$ such that $u \equiv_{P_1 \vee P} v$ and $u' \equiv_{P_2 \vee P} v$. But $Q \leq P_2$ means that $u' \equiv_{P_2 \vee P} v$ implies $u' \equiv_{Q \vee P} v$, and this means that $P_1 \bot Q \vert P$. Similarly, $u \equiv_P u'$ implies the existence of an element $v$ such that $u \equiv_{P_1 \vee P} v$ and $u' \equiv_{P_2 \vee P} v$, says also that $P_1 \bot P_2 \vee P \vert P$, hence C4.
\end{proof}

A three-place relation like $P_1 \bot P_2 \vert P$ satisfying C1 to C4 has been called a \textit{quasi-separoid} (q-separoid) in \cite{kohlas17}. It is a reduct of a separoid, a concept discussed in \cite{dawid01}. Conditional independence structures can be exploited for computational purposes within a set algebra, and as we shall see later also within an information algebra \cite{kohlas17}. The base for this is the next theorem. The issue of conditional independence will be further discussed later in Section \ref{sec:CondIndep}.

\begin{theorem} \label{th:LocCompBas}
Let $(\mathcal{S};\mathcal{E},\cap,U,\emptyset)$ be a set algebra.
\begin{enumerate}
\item If $P_1 \bot P_2 \vert P$ for $P_1,P_2,P \in \mathcal{D}$ and $\sigma_{P_1}(X) = X$, $\sigma_{P_2}(Y) = Y$ for $X,Y \in \mathcal{S}$, then
\begin{eqnarray} \label{eq:CombAxPart}
\sigma_P(X \cap Y) = \sigma_P(X) \cap \sigma_P(Y).
\end{eqnarray}
\item If $P_1 \bot P_2 \vert P$ for $P_1,P_2,P \in \mathcal{D}$ and $\sigma_{P_1}(X) = X$, then
\begin{eqnarray} \label{eq:ExtAxPart}
\sigma_{P_2}(X) = \sigma_{P_2}(\sigma_P(X)).
\end{eqnarray}
\end{enumerate}
\end{theorem}

\begin{proof}
1.) Saturation operators are monotone. Therefore, from $X \cap Y \subseteq X,Y$ it follows that $\sigma_P(X \cap Y) \subseteq \sigma_P(X) \cap \sigma_P(Y)$. Consider now an element $u \in \sigma_P(X) \cap \sigma_P(Y)$. Then, $u \in \sigma_P(X)$ implies that there is an element $u' \in X$ such that $u \equiv_P u'$, hence such that $u,u'$ are together in some block $B'$ of $P$. Also, since $\sigma_{P_1}(X) = X$, we have  $u' \in B_1 \subseteq X$ for some block $B_1$ of $P_1$. In the same way, we have $u \equiv_P u''$ for some element $u'' \in B_2 \cap B''$ for some block $B_2 \subseteq Y$ of $P_2$ and $B''$ of $P$. It follows that $B' = B'' = B$, hence $B_1 \cap B \not= \emptyset$ and $B_2 \cap B \not= \emptyset$. Then $P_1 \bot P_2 \vert P$ implies $B_1 \cap B_2 \cap B \not= \emptyset$ and we have $\emptyset \not= B_1 \cap B_2 \subseteq X \cap Y$. So there is an element $v \in B_1 \cap B_2 \cap B$ such that $v \equiv_P u$ and $v \in X \cap Y$, hence $u \in \sigma_P(X \cap Y)$.

2.) From $X \subseteq \sigma_P(X)$ it follows that $\sigma_{P_2}(X) \subseteq \sigma_{P_2}(\sigma_P(X))$. Consider an element $u \in \sigma_{P_2}(\sigma_P(X))$. Then there is an element $u' \in \sigma_P(X)$ such that $u \equiv_{P_2} u'$. Hence there is a block $B_2$ of $P_2$ containing $u$ and $u'$. Further there is a $u'' \in X$ such that $u' \equiv_P u''$ and there is a block $B$ of $P$ containing $u',u''$. Then we have $B_2 \cap B \not= \emptyset$. Further, since $\sigma_{P_1}(X) = X$, there is a block $B_1 \subseteq X$ of $P_1$ containing $u''$, so that $B_1 \cap B \not= \emptyset$. Then $P_1 \bot P_2 \vert P$ implies that $B_1 \cap B_2 \cap B \not= \emptyset$ and we have $\emptyset \not= B_1 \cap B_2 \subseteq X$. Select an element $v \in B_1 \cap B_2 \cap B$; then $v \equiv_{P_2} u$ and $v \in X$, so that we have $u \in \sigma_{P_2}(X)$.
\end{proof}

We call (\ref{eq:CombAxPart}) and (\ref{eq:ExtAxPart}) the \textit{combination property} and the \textit{extraction property} respectively of the set algebra. We shall see later that these properties induce similar properties in an information algebra.

To conclude this section, an important special class of set algebras will be introduced. Given two partitions $P_1$ and $P_2$ of some universe $U$ with associated saturation operators $\sigma_1$ and $\sigma_2$, define the map
\begin{eqnarray*}
\sigma(X) = \bigcup_{k \in \omega} \sigma_1 \circ \sigma_2 \circ \sigma_3 \circ \sigma_4 \circ \cdots \sigma_k(X),
\end{eqnarray*}
for subsets $X$ of $U$, where $\sigma_k = \sigma_1$, if $k$ is odd, and $\sigma_k = \sigma_2$, if $k$ is even. Clearly, $\sigma(X)$ is the smallest set containing $X$ and which is the union of $P_1$-blocks as well as $P_2$-blocks. It follows that $\sigma = \sigma_{P_1 \wedge P_2}$. If $\sigma_1 \circ \sigma_2 = \sigma_2 \circ \sigma_1$, then, since saturation operators are idempotent, we have $\sigma = \sigma_1 \circ \sigma_2$. In this case we say that partitions $P_1$ and $P_2$ \textit{commute}. 
Note that $P_1$ and $P_2$ commute iff for any block $C$ of  $P_1 \wedge P_2$ and $B_1,B_2 \subseteq C$, where $B_1$ and $B_2$ are blocks of $P_1$ and $P_2$ respectively, it follows that $B_1 \cap B_2 \not= \emptyset$.

Let $(\mathcal{D};\leq)$ be a sublattice of $(Part(U);\leq)$ of commuting partitions. This turns out to be the necessary and sufficient condition that $P_1 \bot P_2 \vert P$ if and only if $(P_1 \vee P) \wedge (P_2 \vee P) = P$.

\begin{theorem} \label{th:CommCondIndep}
Let $(\mathcal{D};\leq)$ be a sublattice of $(Part(U);\leq)$. Then $P_1 \bot P_2 \vert P \Leftrightarrow (P_1 \vee P) \wedge (P_2 \vee P) = P$ iff the partitions in $\mathcal{D}$ pairwise commute.
\end{theorem}

\begin{proof}
Assume first $P_1 \bot P_2 \vert P$. Define $Q = (P_1 \vee P) \wedge (P_2 \vee P)$. Now, by C3 we have $P_1 \vee P \bot P_2 \vee P \vert P$ and from C4 is follows that $Q \bot Q \vert P$, since $Q \leq P_1 \vee P, P_2 \vee P$. So, if $u \equiv_P u'$ there is an element $v$ such that $u \equiv_Q v$ and $u' \equiv_Q v$, hence $u \equiv_Q u'$. But this mean that $Q \leq P$. On the other hand we have $(P_1 \vee P) \wedge (P_2 \vee P) \geq P$, so that $Q = P$. Therefore we have always $P_1 \bot P_2 \vert P \Rightarrow (P_1 \vee P) \wedge (P_2 \vee P) = P$.

Assume then that $(P_1 \vee P) \wedge (P_2 \vee P) = P$ implies $P_1 \bot P_2 \vert P$, so that in particular $P_1 \bot P_2 \vert P_1 \wedge P_2$. Hence, if $C$, $B_1$ and $B_2$ are blocks of partitions $P_1 \wedge P_2$ and $P_1$, $P_2$ respectively, then $B_1 \cap C \not= \emptyset$ and $B_2 \cap C \not= \emptyset$ imply $B_1 \cap B_2 \cap C \not= \emptyset$. But since $P_1 \wedge P_2 \leq P1,P_2$ we have $B_1,B_2 \subseteq C$, then $B_1 \cap B_2 \not= \emptyset$ means that the partitions $P_1$ and $P_2$ commute. 

Conversely, assume that $P_1$ and $P_2$ commute and that $(P_1 \vee P) \wedge (P_2 \vee P) = P$. Consider blocks $B_1$, $B_2$ and $B$ of partitions $P_1$, $P_2$ and $P$ respectively, and, further, let $B'_1 = B_1 \cap B \not= \emptyset$ and $B'_2 = B_2 \cap B \not= \emptyset$ such that $B'_1$ and $B'_2$ are blocks of partitions $P_1 \vee P$ and $P_2 \vee P$, and both subsets of $B$. Then, since $P_1$ and $P_2$ commute, we have $B'_1 \cap B'_2 \not= \emptyset$ hence $B_1 \cap B_2 \cap B \not= \emptyset$, and so indeed $P_1 \bot P_2 \vert P$.
\end{proof}

Consider now a set algebra $(\mathcal{S};\mathcal{E},\cap,\emptyset,U)$, where $\mathcal{E}$ is the set of saturation operators from a sublattice $(\mathcal{D};\leq)$ of commuting partitions. Then $\mathcal{E}$ is closed under composition, since 
\begin{eqnarray*}
\sigma_{P_1} \circ \sigma_{P_2} = \sigma_{P_2} \circ \sigma_{P_1} = \sigma_{P_1 \wedge P_2}
\end{eqnarray*}
for any pair of partitions $P_1$ and $P_2$ from $\mathcal{D}$. So, $(\mathcal{E};\circ)$ is in this case an idempotent commutative semigroup under composition. Therefore, the set algebra $(\mathcal{S};\mathcal{E},\cap,\emptyset,U)$ is called \textit{commutative} in this case. This gives rise to an important special class of information algebras, as we shall see. 

The following is an important example of a a commutative set algebra.

\textit{Multivariate Algebras}

In many applications a set of variables is considered and the information one is interested in concerns the values of certain groups of variables, similar to ordinary relational algebra in database theory, see \cite{kohlas03} for more general relational information algebras. So, consider a countable family of variables $ X = \{X_i: i\in \mathbb{N}\}$,  and let $V_{i}$ denote the set of possible values of the variable $X_i$. For a subset $s\subseteq X$ of variables define
\begin{eqnarray}
V_s = \prod_{X_i \in s} V_i
\nonumber
\end{eqnarray}
to be the set of possible answers relative to $s$.
Let
\begin{eqnarray}
V_{\omega} = \prod_{i=1}^{\infty} V_{i}.
\nonumber
\end{eqnarray}
and let $\Psi = 2^{V_\omega}$, the power set of $V_{\omega}$. Note that the elements of $V_{\omega}$ are the sequences $t = (t_{1},t_{2},\ldots)$ with $t_{i} \in V_{i}$. An element $\phi \in \Psi$ may be interpreted as a piece of information, which states that a generic element $t\in V_{\omega} $ belongs to the set $\phi$. Within $\Psi$ we define combination by  set intersection, which represents aggregation of information:
\begin{eqnarray}
\phi \cdot \psi = \phi \cap \psi.
\nonumber
\end{eqnarray}
Equipped with this operation, $\Psi$ is an idempotent commutative semigroup with least element $V_{\omega}$ and greatest element $\emptyset$ under the associated information order (given by $\psi \leq \phi$ iff $\phi \subseteq \psi$). The smaller the subset representing a piece of information about elements of $V_{\omega}$, the more information it contains.

Let $s$ be any subset of $X$. Define, for any sequence $t$ in $V_{\omega}$, its {\em restriction} to $ s$, denoted by  $t \vert s$, as follows: If $s = \{X_{i_1},X_{i_2},\dots\}$, then  $t \vert s = (t_{i_1},t_{i_2},\dots)$. Also, define an equivalence relation $\equiv_s$ in $V_{\omega}$ by
\begin{eqnarray}
t \equiv_{s} t' \textrm{ iff}\ t \vert s = t' \vert s.
\nonumber
\end{eqnarray}
It is easy to see that the relational product $\equiv_s\ \bowtie\ \equiv_{s'}$ of two such equivalence relations is $ \equiv_{s\cap s'}$. It follows that any two of such equivalence relations commute, and thus so do their associated partitions $P_s$ of $V_{\omega}$, as well as the saturation operators $\sigma_s$ associated with $P_s$. Let $\mathcal{F}$ be the set of all partitons $P_s$ ($s\subseteq X$) and write $S_{\mathcal{F}}$ for the set of all saturation operators $\sigma_s$  associated with the partitions $P_s\in \mathcal{F}$.

It is immediate that $\sigma_s$ maps $\Psi$ into $\Psi$ and that $\sigma_s(\emptyset) = \emptyset$ for all $s\subseteq X$. So $(\Psi;S_{\mathcal{F}},\cdot,V_{\omega},\emptyset,\circ)$ is a commutative set algebra. It is called a \textit{multivariate} information algebra; also,  the sets $\sigma_{s}(\phi)$ are called \textit{cylindric} over $s$.

We are going to show that any abstract information algebra is in some sense part of or contained in a set algebra commutative or not, or, more precisely,  \textit{embedded} into a set algebra. But before we present two further examples of information algebras.

\section{Some Examples of Information Algebras} \label{subsc:ExplInfAlg}

%
%

\subsection{Algebra of Strings}

Consider a finite alphabet $\Sigma$, the set $\Sigma^{*}$ of finite strings over $\Sigma$, including the empty string $\epsilon$, and the set $\Sigma^{\omega}$ of infinite strings over $\Sigma$. Let $\Sigma^{**} = \Sigma^{*} \cup \Sigma^{\omega} \cup \{0\}$, where $0$ is a symbol not contained in $\Sigma$. For two strings $r,s \in \Sigma^{**}$, define $r \leq s$, if $r$ is a prefix of $s$ or if $s = 0$. The empty string is a prefix of any of any string. Define a combination operation in $\Sigma^{**}$ as follows:

\begin{eqnarray}
r \cdot s = \left\{ \begin{array}{ll} s, & \textrm{if}\ r \leq s, \\ r & \textrm{if}\ s \leq r, \\ 0, & \textrm{otherwise}. \end{array} \right.
\nonumber
\end{eqnarray}

Clearly, $(\Sigma^{**},\cdot)$ is a commutative idempotent semigroup. The empty string $\epsilon$ is the unit element, and the adjoined element $0$ is the null element of combination. For extraction, we define operators $\epsilon_{n}$ for any $n \in \mathbb{N}$ and also for $n = \infty$. Let $\epsilon_{n}(s)$ be the prefix of length $n$ of string $s$, if the length of $s$ is at least $n$, and let $\epsilon_{n}(s) = s$ otherwise. In particular, define $\epsilon_{\infty}(s) = s$ for any string $s$ and $\epsilon_{n}(0) = 0$ for any $n$. It is easy to verify that any $\epsilon_{n}$ maps $\Sigma^{**}$ into itself, and that it satisfies conditions E1 to E5 for an extraction operator. So, the so-called \textit{string algebra} $(\Sigma^{**};\mathcal{E},\cdot,\epsilon,0)$ is an instance of an  information algebra. Moreover, in this example, as in many others, $\mathcal{E} = \{\epsilon_{n}: n \in \mathbb{N} \cup \{\infty\}\}$ is a commutative and idempotent semigroup under composition of maps. Algebras with this particular property will be called commutative, just like commutative set algebras (see Section \ref{subsec:CommAlg}).

\subsection{Lattice-Valued Algebras} \label{subsec:LattValAlg}

Consider a family of partitions $P_x$, $x \in D$, of an universe $U$, such that all the blocks of each partition $P_x$ have only a finite number of elements and let $\Lambda$ be a bounded distributive lattice with least element $\bot$ and greatest element $\top$. Consider a set $\Psi$ of maps $\psi : U \rightarrow \Lambda$. We assume that for each map $\psi \inÊ\Psi$, there is a partition $P_x$ such that $\psi(u) = \psi(v)$ whenever $u \equiv_x v$, that is, if $u$ und $v$ belong to a same block of $P_x$. Further, we assume that the family of partitions $P_x$ forms a join-semilattice for $x \in D$. Within $\Psi$,  define an operation of \textit{combination} $\phi \cdot \psi$ by
\begin{eqnarray}
(\phi \cdot \psi)(u) = \phi(u) \wedge \psi(u) \textrm{ for all}\ u \in U.
\nonumber
\end{eqnarray}
Clearly, under this combination operation, $\Psi$ becomes a semigroup with unit element $1$ defined by $1(u) = \top$ for all $u \in U$ and null element $0$ defined by $0(u) = \bot$ for all $u$. Note that the information order $\phi \leq \psi$ in $\Psi$ is given by  $\phi(u) \geq \psi(u)$ for all $u \in U$, hence the information order is related to the opposite order in the lattice $(\Lambda;\leq)$. In the information algebra combination is the join in information order. There exists also the infimum for any two elements $\phi$ and $\psi$ of $\Psi$, and in fact
\begin{eqnarray*}
(\phi \wedge \psi)(u) = \phi(u) \vee \psi(u).
\end{eqnarray*}
So, $(\Psi,\leq)$ is a bounded lattice, and in fact a distributive one, since the distributive laws are inherited form the distributive lattice $\Lambda$. 

For any partition $P_x$ we introduce an operator $\epsilon_x$ mapping $\Psi$ into $\Psi$, defined by
\begin{eqnarray}
\epsilon_x(\psi)(u) = \vee\{\psi(v): v \equiv_x u\}.
\nonumber
\end{eqnarray}
Here we use the finiteness of blocks of the partitions $P_x$ which guarantees that the supremum always exists. Using the distributivity of the lattice $\Lambda$ it is easy to verify that all of these operators $\epsilon_x$ are existential quantifiers, that is satisfy conditions E1 to E3 (see Section \ref{sec:FomFreeAlg}).

Further, E4 holds by assumption. And if $P_x \leq P_y$, then $u \equiv_y v$ implies $u \equiv_x v$ so that if $\epsilon_x(\psi) = \psi$, we have
\begin{eqnarray*}
\epsilon_y(\psi)(u) = \vee\{\psi(v):v \equiv_y u\} = \vee\{\psi(v):v \equiv_x u\} = \psi(u).
\end{eqnarray*}
This shows that condition E5 is satisfied too. Thus $(\Psi;\mathcal{E},\cdot,1,0)$, where $\mathcal{E}$ is the set of all operators $\epsilon_x$ for all partitions $P_x$ in the join-semilattice, is an information algebra. Moreover, by the associative law it follows that for any $u$ in $U$
\begin{eqnarray*}
\epsilon_x(\phi \wedge \psi)(u) = \epsilon_x(\psi)(u) \wedge \epsilon_x(\phi)(u),
\end{eqnarray*}
hence we have $\epsilon_x(\phi \wedge \psi) = \epsilon_x(\phi) \wedge \epsilon_x(\psi)$. This means that the information algebra is a in fact a distributive lattice information algebra as considered in Section \ref{subsec:DistLattInfAlg} below. Moreover, it is a particular case of a semiring induced valuation algebra, see \cite{kohlaswilson06}

There are many more instances  of information algebras related to algebraic logic, graph theory,  linear algebra and convex sets, and other topics as well. We refer to \cite{kohlas03,poulykohlas11,kohlasschmid14} for further examples.


\section{Information and Set Algebras} \label{sec:InfAndSetAlg}

\subsection{Order-Generating Sets} \label{subsec:OrdGenSets}

We are going to associate  abstract information algebras $(\Psi;\mathcal{E},\cdot,1,0)$ where $\mathcal{E} = \{\epsilon_x:x \in D)$ and $(D;\leq)$ a join semilattice with some appropriate set algebras. The key notion is the one of an order-generating set.

\begin{definition}
\textit{Order-Generating Set:} A subset $X$ of $\Psi$ not containing the null element $0$ in an information algebra $(\Psi;\mathcal{E},\cdot,1,0)$ is called oder-generating, iff for all $\psi \in \Psi$, $\psi \not= 0$,
\begin{eqnarray*}
\psi = \inf\{\uparrow\!\psi \cap X\}.
\end{eqnarray*}
\end{definition}

Here $\uparrow\!\psi$ is the set of all $\phi \in \Psi$ which are greater (more informative) than $\psi$. Note that $\Psi/\{0\}$ is an order generating set. Further, and more interesting examples of order-generating sets will be given in Section \ref{sec:Exmpls} below. Order-generating sets are also sometimes called meet-dense sets \cite{daveypriestley97}. We give here a well-known result on a characterization of order-generatings sets \cite{gierz03}.

\begin{theorem} \label{th:CharOrdGen}
A subset $X$ of the join-semilattice $(\Psi;\leq)$ is order generating iff $\phi \not\leq \psi$ implies that there is a $\chi \in X$ such that $\psi \leq \chi$ and $\phi \not\leq \chi$.
\end{theorem}

\begin{proof}
If $X$ is order generating then $\psi = \inf P$ for some subset $P$ of $X$. But then $\phi \not\leq \psi$ implies that there is a $\chi \in P$ such that $\phi \not\leq \chi$. Conversely $\psi$ is a lower bound of $\uparrow\!\psi \cap X$. Assume $\phi$ to be another lower bound of this set. We claim that $\phi \leq \psi$. Assume on the contrary that $\phi \not\leq \psi$. But then there is by assumption a $\chi \in \uparrow\! \psi \cap X$ so that $\psi \leq \chi$ but $\phi \not\leq \chi$ which contradicts the assumption that $\phi$ is a lower bound of $\uparrow\!\psi \cap X$. Therefore $\phi \leq \psi$ and $\psi = \inf \uparrow\!\psi \cap X$.
\end{proof}

We consider the universe $X$ and define a mapping $f$ of the information algebra $\Psi$ into the power set of $2^X$ by
\begin{eqnarray*}
\psi \mapsto f(\psi) = \uparrow\!\psi \cap X.
\end{eqnarray*}

Now, obviously, $f$ is a join-homomorphism of $(\Psi;\leq)$ into $(2^X;\subseteq)^\vartheta$. So combination (or join) in $\Psi$ maps to intersection in $2^X$, the null element to the emptyset and the unit to $X$.
\begin{eqnarray*}
f(\phi \cdot \psi) &=& \uparrow\!(\phi \cdot \psi) \cap X = (\uparrow\!\phi \cap X) \cap ( \uparrow\!\psi \cap X) = f(\phi) \cap f(\psi), \\
f(1) &=& X, \\
f(0) &=& \emptyset.
\end{eqnarray*}
Furthermore, $f$ is injective.

Next, we associate the extraction operators $\epsilon_x$ in $\mathcal{E}$ with some saturation operators of partitions in $X$. For this purpose consider in $X$ the equivalence relations
\begin{eqnarray*} 
\alpha \equiv_x \beta \textrm{ iff}\ \epsilon_x(\alpha) = \epsilon_x(\beta), \textrm{ where}\ \alpha,\beta \in X.
\end{eqnarray*} 
The blocks of the associated partitions $P_x$ of $\Psi$ are determined by $\{\epsilon_x(\beta):\beta \in X\} = \epsilon_x(X)$. We now make for the rest of the paper the additional assumption that
\begin{eqnarray} \label{eq:Infomorph}
x \not= y \Rightarrow \epsilon_x(\Psi) \not= \epsilon_y(\Psi).
\end{eqnarray}
Informally, different domains represent different answers. Let $g$ be the map from domains $D$ to partitions of $\Psi$ defined by
\begin{eqnarray*} 
g(x) = P_x, \textrm{ where}\ P_x \textrm{ is the partition of $X$ associated with}\ \equiv_x.
\end{eqnarray*}
The image $g(D)$ of $D$ is the family of partitions $\mathcal{P}(D) = \{P_x:x \in D\}$ and by the assumption above $g$ is injective and also surjective on $\mathcal{P}(D)$. It is furthermore a join-hommorphism,
\begin{eqnarray*}
x \leq y \Rightarrow P_x \leq P_y, \quad g(x \vee y) = P_{x \vee y} = P_x \vee P_y.
\end{eqnarray*}
So, $(D;\leq)$ and $(\mathcal{P}(D);\leq)$ are isomorphic join-semilattices. This induces also a mapping $g$ from $\mathcal{E}$ into the set of saturation operators $\sigma_{P_x}$ associated with the partitions $P_x$ in $\mathcal{P}(D)$,
\begin{eqnarray*}
g(\epsilon_x) = \sigma_{g(x)} = \sigma_{P_x}.
\end{eqnarray*}

For the pair of maps $(f,g)$ to be a homomorphism between the information algebra $(\Psi;\mathcal{E},\cdot,1,0)$ and the set algebra $(2^X;g(\mathcal{E}),\cap,\Psi,\emptyset)$, we need in addition that 
\begin{eqnarray*}
f(\epsilon_x(\psi)) = \sigma_{g(x)}(f(\psi)).
\end{eqnarray*}
For this condition to hold, we require an additional condition for the order-generating set, as given in the following definition:

\begin{definition}
\textit{Strongly Order-Generating Sets:} A subset $X$ not containing the null element $0$ of $\Psi$ in an information algebra $(\Psi;\mathcal{E},\cdot,1,0)$ is called strongly oder-generating, if it is order-generating, and whenever $\epsilon_x(\alpha) \geq \epsilon_x(\psi)$, for $\alpha \in X$, there is a $\gamma \in X$ such that $\alpha \equiv_x \gamma$ and $\gamma \geq \psi$.
\end{definition}

Now we have the following general representation theorem for information algebras.

\begin{theorem} \label{th:GenRep_1}
Let $(\Psi;\mathcal{E},\cdot,1,0)$ be an information algebra and $X \subseteq \Psi$ a strongly order-generating set. Then the pair of maps $(f,g)$ as defined above define an embedding of the information algebra $(\Psi;\mathcal{E},\cdot,1,0)$ into the set algebra $(2^X;g(\mathcal{E}),\cap,\Psi,\emptyset)$.
\end{theorem}

\begin{proof}
It remains only to show that
\begin{eqnarray*}
\uparrow\!\epsilon_x(\psi) \cap X = \sigma_{P_x}(\uparrow\!\psi \cap X).
\end{eqnarray*}
Consider first  $\alpha \in \sigma_{P_x}(\uparrow\!\psi \cap X)$. Then there is an element $\beta$ in $X$ such that $\beta \equiv_x \alpha$ and $\beta \geq \psi$. But then $\alpha \geq \epsilon_x(\alpha) = \epsilon_x(\beta) \geq \epsilon_x(\psi)$, hence $\alpha \in \uparrow\!\epsilon_x(\psi) \cap X$.

Conversely, let $\alpha \in\ \uparrow\!\epsilon_x(\psi) \cap X$ such that $\alpha \geq \epsilon_x(\psi)$. Then, because $X$ is strongly order-generating, there is an element $\beta \in X$ such that $\alpha \equiv_x \beta$ and $\beta \geq \psi$. So, $\alpha \in \sigma_{P_x}(\uparrow\!\psi \cap X)$. This proves then that $\uparrow\!\epsilon_x(\psi) = \sigma_{P_x}(\uparrow\!\psi)$.
\end{proof}

This is a first general representation theorem, whose application to concrete situations will be presented in Section \ref{sec:Exmpls}. In the next section a generalization of this representation theorem will be discussed.

\subsection{Locally Order-Generating Sets} \label{subsec:LocOrdGenSet}

Let $X$ be an order-generating set of an information algebra $(\Psi;\mathcal{E},\cdot,1,0)$. Then, we claim that $\epsilon_x(X)$ is order-generating in $\epsilon_x(\Psi)$. In fact, consider
\begin{eqnarray*}
\inf\{\alpha:\alpha \in \epsilon_x(X),\alpha \geq \epsilon_x(\psi)\} \geq \epsilon_x(\psi).
\end{eqnarray*}
But, if $\alpha \in X$ such that $\alpha \geq \epsilon_x(\psi)$, then $\epsilon_x(\alpha) \in \epsilon_x(X)$ and $\epsilon_x(\alpha) \geq \epsilon_x(\psi)$. Therefore,
\begin{eqnarray*}
\lefteqn{\epsilon_x(\psi) = \inf\{\alpha \in X:\alpha \geq \epsilon_x(\psi)\} \geq \inf\{\epsilon_x(\alpha),\alpha \in X,\alpha \geq \epsilon_x(\psi)\}} \\
&&= \inf \{\chi:\chi \in \epsilon_x(\Psi),\chi \geq \epsilon_x(\psi)\}.
\end{eqnarray*}
Therefore, we have
\begin{eqnarray*}
\inf\{\chi:\chi \in \epsilon_x(X),\chi \geq \epsilon_x(\psi)\} = \epsilon_x(\psi)
\end{eqnarray*}
which shows that the set $\epsilon_x(X)$ is order-generating in $\epsilon_x(\Psi)$,
\begin{eqnarray*}
\epsilon_x(\psi) = \inf\{\uparrow\!\epsilon_x(\psi) \cap \epsilon_x(X)\}.
\end{eqnarray*}
We say that the sets $\epsilon_x(X)$ are \textit{locally order-generating}. 

Let $T = \bigcup_{x \in D} \epsilon_x(X)$ and define for $t \in T$ a labeling function $d(t)$ by $d(t) = x$ if $t \in \epsilon_x(X)$ and define further a partial function $\pi_x(t)$ for $x \leq d(t)$ by $\pi_x(t) = \epsilon_x(t)$. Then, if $X$ is strongly order-generating, we have the following properties:
\begin{enumerate}
\item If $x \leq d(t)$, then $d(\pi_x(t)) = x$,
\item if $x \leq y \leq d(t)$, then $\pi_x(\pi_y(t)) = \pi_x(t)$,
\item if $d(t) = x$, then $\pi_x(t) = t$,
\item if $d(t) = x$, $x \leq y$, then there is a $s \in T$ such that $d(s) = y$ and $\pi_x(s) = t$,
\item if $x \leq y$, $t \in T$ with $d(t) = x$ and $\psi \in \Psi$ with support $y$ such that $t \geq \epsilon_x(\psi)$, then there is an element $s \in T$ with domain $y$ and such that $\pi_x(s) = t$ and $s \geq \psi$.
\end{enumerate}
Such a system is called a \textit{tuple system}; a particular variant of such a system has been introduced in \cite{kohlas03}.

We may now also consider a family $X_x \subseteq \epsilon_x(\Psi)$ of locally order-generating sets for $x \in D$, so that for all $\psi \in \Psi$
\begin{eqnarray*}
\epsilon_x(\psi) = \inf\{\uparrow\!\epsilon_x(\psi) \cap X_x)\}.
\end{eqnarray*}
We require $T = \bigvee_{x \in D} X_x$ to be a tuple system, with labeling and projection functions as defined above. As we have seen, a strongly order-generating set $X$ in $\Psi$ induces such a tuple system, the converse however does not hold; there is, in general, no order-generating subset $X$ in $\Psi$ inducing such a family of locally order-generating systems.

We consider now, relative to a tuple system $X_x$ for $x \in D$, maps
\begin{eqnarray*}
a : D \rightarrow \bigcup_{x \in D} X_x, \textrm{ such that}\ a_x \in X_x,
\end{eqnarray*}
and we require for these maps that the additional consistency condition that $x \leq y$ implies $a_x = \pi_x(a_y)$. This implies also that any system of components $a_{x_1},\ldots a_{x_n}$ of a consistent $a$ are compatible,
\begin{eqnarray*}
a_{x_1} \cdot \ldots \cdot a_{x_n} = \epsilon_{x_1}(a_{x_1 \vee \ldots x_n}) \cdot \ldots \cdot \epsilon_{x_n}(a_{x_1 \vee \ldots x_n}) \leq a_{x_1 \vee \ldots x_n} \not= 0.
\end{eqnarray*}
Let $U$ be the set of all consistent maps $a$. We define for any $\psi \in \Psi$ the sets $X_x(\psi) = \uparrow\!\epsilon_x(\psi) \cap X_x$, and 
\begin{eqnarray*}
\mathbf{X}(\psi) = \{a \in U:a_x \in X_x(\psi) \textrm{ for all}\ x \in D\}.
\end{eqnarray*}
Note that if $a_x \in X_x(\psi)$ for some $x \in D$, then $a_y \in X_y(\psi)$ for any other $y \in D$. In fact, we have $\epsilon_x(\psi) \leq a_x = \epsilon_x(a_{x \vee y})$ and $\epsilon_y(\psi) \leq \epsilon_y(\epsilon_x(a_{x \vee y})) \leq \epsilon_y(a_{x \vee y}) = a_y$.

As with order-generating sets, we define a map $f : \Psi \rightarrow U$ by
\begin{eqnarray*}
\psi \mapsto f(\psi) = \mathbf{X}(\psi).
\end{eqnarray*}
This is again a join-homomorphism for $(\Psi;\leq)$ into $(2^U;\subseteq)^\vartheta$:

\begin{lemma} \label{le:LocOrdSetHomom}
For the map $f$ as defined above we have
\begin{enumerate}
\item $f(\phi \cdot \psi) = \mathbf{X}(\phi) \cap \mathbf{X}(\psi)$,
\item $f(0) = \emptyset$,
\item $f(1) = U$.
\end{enumerate}
\end{lemma}

\begin{proof}
The last two statements are obvious. If $a \in \mathbf{X}(\phi \cdot \psi)$, then for any $x$ in $D$, $a_x \in At_x(\phi \cdot \psi)$, hence 
\begin{eqnarray*}
a_x \geq \epsilon_x(\phi \cdot \psi) \geq \epsilon_x(\phi) \cdot \epsilon_x(\psi).
\end{eqnarray*}
So, we see that $a_x \in X_x(\phi)$ and $a_x \in X_x(\psi)$, hence $a \in \mathbf{X}(\phi) \cap \mathbf{X}(\psi)$. Conversely if $a \in \mathbf{X}(\phi) \cap \mathbf{X}(\psi)$, then we have $a_x \geq \phi$ and $a_y \geq \psi$, if $x$ and $y$ are supports of $\phi$ and $\psi$ respectively. Then we have also $a_{x \vee y} \geq \epsilon_{x \vee y}(\phi)$ and $a_{x \vee y} \geq \epsilon_{x \vee y}(\psi)$ and therefore
\begin{eqnarray*}
a_{x \vee y} \geq \epsilon_{x \vee y}(\phi) \cdot \epsilon_{x \vee y}(\psi) = \phi \cdot \psi,
\end{eqnarray*}
hence $a \in \mathbf{X}(\phi \cdot \psi)$, which concludes the proof.
\end{proof}

Next we define in $U$ equivalence relations $a \equiv_x b$ iff $a_x = b_x$ and denote the corresponding partitions in $U$ by $P_x$, which have the saturation operators
\begin{eqnarray*}
\sigma_{P_x}(S) = \{a \in U:\exists b \in U \textrm{ so that}\ b \in S, a \equiv_x b\}.
\end{eqnarray*}
The map $g : D \rightarrow Part(U)$ defined by $g(x) = P_x$ is injective, if we again assume (\ref{eq:Infomorph}).

As before, it remains to show that 
\begin{eqnarray*}
f(\epsilon_x(\psi)) = \sigma_{P_x}(f(\psi)),
\end{eqnarray*}
if we want to show that the information algebra $(\Psi;\mathcal{E}),\cdot,1,0)$ is embedded in the set algebra $(2^U;g(\mathcal{E}),\cap,U,\emptyset)$. Assume first, that
$a \in \mathbf{X}(\epsilon_x(\psi))$ so that $a_x \geq \epsilon_x(\psi)$. Let $\psi$ have support $y$. Then by property 5 of a tuple system there is a tuple $s$ with domain $x \vee y$ such that $\pi_x(s) = t$ and $s \geq \psi$. Let $b \in U$ such that $b_x = s$, then $a \equiv_x b$ and we see that $a \in \sigma_{P_x}(\mathbf{X}(\psi))$. And if, conversely, $a \in \sigma_{P_x}(\mathbf{X}(\psi))$, then there is a $b \in \mathbf{X}(\psi)$ such that $a_x = b_x$. If $y$ is a support of $\psi$, then $x \vee y$ is also a support of $\psi$, hence $b_{x \vee y} \geq \psi$ and therefore $a_x = b_x = \pi_x(b_{x \vee y}) \geq \epsilon_x(\psi)$, hence $a \in \mathbf{X}(\epsilon_x(\psi))$. This proves that
\begin{eqnarray*}
\mathbf{X}(\epsilon_x(\psi)) = \sigma_{P_x}(\mathbf{X}(\psi)).
\end{eqnarray*}
This completes the proof of the following new representation theorem.

\begin{theorem} \label{th:GenRep_2}
Let $(\Psi;\mathcal{E},\cdot,1,0)$ be an information algebra and $X_x$ for $x \in D$ a tuple system of locally order-generating sets. Then the pair of maps $(f,g)$ as defined above define an embedding of the information algebra $(\Psi;\mathcal{E},\cdot,1,0)$ into the set algebra $(2^U;g(\mathcal{E}),\cap,U,\emptyset)$.
\end{theorem}

These abstract general results shall be applied in the next Section to various cases generating thus different representation theorems for information algebras.


\section{Examples of Embeddings} \label{sec:Exmpls}

\subsection{General Information Algebra}

In any information algebra $(\Psi;\mathcal{E},\cdot,1,0)$ the set $X = \Psi/\{0\}$ is order generating, since
\begin{eqnarray*}
\psi = \inf \uparrow\!\psi/\{0\}.
\end{eqnarray*}
It is even strongly order-generating as the following lemma shows.

\begin{lemma} \label{le:StrongOrdGen}
Assume for $\phi,\psi \in \Psi$, different from 0, that $\epsilon_x(\phi) \geq \epsilon_x(\psi)$. Then there is an element $\chi \in \Psi$ such that $\psi \leq \chi$ and $\epsilon_x(\chi) = \epsilon_x(\phi)$.
\end{lemma}

\begin{proof}
Let $\chi = \epsilon_x(\phi) \cdot \psi$. Then $\psi \leq \chi$ and $\epsilon_x(\chi) = \epsilon_x(\phi) \cdot \epsilon_x(\psi) = \epsilon(\phi)$.
\end{proof}

So, according to Theorem \ref{th:GenRep_1}, the information algebra $(\Psi;\mathcal{E},\cdot,1,0)$ is embedded into the set algebra $(2^{\Psi/\{0\}};g(\mathcal{E}),\cap,\Psi/\{0\},\emptyset)$. The embedding maps are given by
\begin{eqnarray*}
\psi \mapsto \uparrow\!\psi/\{0\}, \quad x \mapsto P_x,
\end{eqnarray*}
where $P_x$ is the partition of $\Psi/\{0\}$ associated with the equivalence relation $\phi \equiv_x \psi$ iff $\epsilon_x(\phi) = \epsilon_x(\psi)$. So we have the following theorem:

\begin{theorem}
An information algebra $(\Psi;\mathcal{E},\cdot,1,0)$ is embedded in the set algebra $(2^{\Psi/\{0\}};g(\mathcal{E}),\cap,\Psi/\{0\},\emptyset)$. It is isomorphic to the sub-setalgebra of upsets $ \uparrow\!\psi/\{0\}$.
\end{theorem}

This is a very general representation theorem. The embedding has an information-theoretic interpretation: A piece of information $\psi$ may be considered as a partial information (to whatever question considered). All elements $\phi \geq \psi$ are in this sense possible (although possibly still incomplete) completions of $\psi$, and $\uparrow\!\psi/\{0\}$ is the consistent selection of all these completions. For example, in the string algebra, a finite string is extended by all strings which have it as a prefix. The sets of strings form a set algebra into which the string algebra is embedded. This point of view will be enforced by the examples in the following sections. Note also that the algebra of all upsets $U \subseteq \Psi/\{0\}$, that is sets, such that $\psi \in U$ and $\psi \leq \phi$, form a set algebra, a subalgebra of $(2^{\Psi/\{0\}};g(\mathcal{E}),\cap,\Psi/\{0\},\emptyset)$. This follows since upsets are closed under intersections and unions and since for any upset
\begin{eqnarray*}
U = \bigcup_{\psi \in U} \uparrow\!\psi/\{0\},
\end{eqnarray*}
and therefore
\begin{eqnarray*}
\sigma_{P_x}(U) = \bigcup_{\psi \in U} \sigma_{P_x}(\uparrow\!\psi/\{0\}) = \bigcup_{\psi \in U} \uparrow\!\epsilon_x(\psi)/\{0\}.
\end{eqnarray*}

And $(\Psi;\mathcal{E},\cdot,1,0)$ is also embedded into this set-algebra, whose elements have a similar interpretation as a consistent family of completions of partial information.

\subsection{Atomic Information Algebras} \label{sec:AtomInfAlg}

\subsubsection{Atomistic Algebras} \label{subsec:AtomisticInfAlg}

In many information algebras there are maximal elements different from $0$. And in many important cases these maximal elements determine the information algebra fully. In this case we shall show that the algebra is still, as in the general case, essentially a set algebra but of a different type. This section is an extension of material developed in \cite{kohlas03} for labeled information algebras and in \cite{kohlasschmid16} for commutative information algebras. It forms also the base for the representation theorems of Boolean information algebra, Section \ref{sec:BooleInfAlg}.

Consider a generalized information algebra $(\Psi;\mathcal{E},\cdot,1,0)$ where the extraction operators $\epsilon_x \in \mathcal{E}$ as usual are indexed by a join-semilattice $(D;\leq)$. Then we define the concept of an atom in $\Psi$ as follows:

\begin{definition}
\textit{Atom:} An element $\alpha \in \Psi$ is 
called an atom, iff
\begin{enumerate}
\item $\alpha \not= 0$,
\item for all $\psi \in \Psi$, if $\alpha \leq \psi$ then either $\alpha = \psi$ or $\psi = 0$.
\end{enumerate}
\end{definition}
So atoms are maximal elements in $(\Psi;\leq)$ different from $0$ \footnote{In order theory atoms are defined as minimal elements. But for our purposes defining atoms as maximally informative elements makes sense. So, our atoms are co-atoms in order theory.}. Here follow a few useful, elementary properties of atoms.

\begin{lemma} \label{le:AtomProp}
Let $(\Psi;\mathcal{E},\cdot,1,0)$ be an information algebra. Then 
\begin{enumerate}
\item if $\alpha$ is an atom, then for all $\psi \in \Psi$ either $\alpha \cdot \psi = 0$ or $\alpha \cdot \psi = \alpha$,
\item if $\alpha$ is an atom, then for all $\psi \in \Psi$ either $\alpha \cdot \psi = 0$ or $\psi \leq \alpha$,
\item if $\alpha$ and $\beta$ are atoms, then either $\alpha = \beta$ or $\alpha \cdot \beta = 0$.
\end{enumerate}
\end{lemma}

\begin{proof}
1.)  and 2.) We have $\alpha \leq \alpha \cdot \psi$. If $\alpha$ is an atom, then by definition either $\alpha \cdot \psi = \alpha$ or $\alpha \cdot \psi = 0$. In the first case we have  $\psi \leq \alpha$.

3.) Here we have $\alpha \leq \alpha \cdot \beta$. So, if $\alpha$ is an atom then either $\alpha = \alpha \cdot \beta$ or $\alpha \cdot \beta = 0$. In the first case it follows that $\beta \leq \alpha$. But because $\beta$ is also an atom and $\alpha \not= 0$, this implies $\alpha = \beta$.
\end{proof}

We denote the set of all atoms of an information algebra $\Psi$ by $At(\Psi)$. According to Lemma \ref{le:AtomProp} any element $\psi \in \Psi$ is either contradictory to an atom $\alpha$ ($\alpha \cdot \psi = 0$) or is implied by an atom $\alpha$ ($\psi \leq \alpha$). Let $At(\psi) = \{\alpha \in At(\Psi):\psi \leq \alpha\}$ be the set of all atoms implying $\psi$. We introduce now particular classes of information algebras.

\begin{definition}
\textit{Atomic and Atomistic Information Algebras:} Let $(\Psi;\mathcal{E},\cdot,1,0)$ be an information algebra. Then,
\begin{enumerate}
\item if for all $\psi \in \Psi$, $\psi \not= 0$, the set $At(\psi)$ is not empty, the algebra is called atomic,
\item if or all $\psi \in \Psi$, $\psi \not= 0$,
\begin{eqnarray*}
\psi = \inf  At(\psi)
\end{eqnarray*}
the algebra is called atomistic,
\item if the algebra is atomistic and for all subsets $X$ of $At(\Psi)$ the infimum $\inf X$ exists and belongs to $\Psi$, then the algebra is called completely atomistic.
\end{enumerate}
\end{definition}

So, in an atomistic information algebra the set $At(\Psi)$ of atoms is order-generating. Here the map $f : \Psi \rightarrow 2^{At(\Psi)}$ defined by $f(\psi) = At(\psi)$ is according to the general scheme of order-generating sets a join-homomorphism between $(\Psi;\leq)$ and $(2^{At(\Psi)};\subseteq)^\vartheta$,
\begin{eqnarray*}
f(\phi \cdot \psi) &=& At(\phi \cdot \psi) = At(\phi) \cap At(\psi) = f(\phi) \cap f(\psi), \\
f(0) &=& \emptyset, \\
f(1) &=& At(\Psi).
\end{eqnarray*}

We may consider atoms in $At(\psi)$ as complete possible completions of $\psi$. In a sense they represent possible worlds, and in an atomistic information algebra, any piece of  information is given by a set of possible worlds

The set $At(\Psi)$ is strongly order-generating in atomistic information algebras as we shall see. In order to show this and also to extend the concept of atomic information algebras, we introduce the notion of atoms relative to a domain $x$. Note that the notion of an atom is only related to the idempotent commutative semigroup $(\Psi;\cdot)$ and has so far no relation to extraction operators. We introduce a now a similar concept which is related to domains $x \in D$ of an information algebra.

\begin{definition}
\textit{Relative Atoms:} Let $(\Psi;\mathcal{E},\cdot,1,0)$ be an information algebra. Then an element $\alpha \in \Psi$ is called an atom relative to $x \in D$, if
\begin{enumerate}
\item $\alpha = \epsilon_x(\alpha) \not= 0$,
\item for all $\psi \inÊ\Psi$, if $\psi = \epsilon_x(\psi) \geq \alpha$ then either $\alpha = \psi$ or $\psi = 0$.
\end{enumerate}
\end{definition}
So, atoms relative to a domain $x$ are maximally informative elements in domain $x$, that is supported by domain $x$. In fact, they are atoms in the subalgebra $(\epsilon_x(\Psi);\mathcal{E}_x,\cdot,1,0)$ of $(\Psi;\mathcal{E},\cdot,1,0)$. Relative atoms have therefore similar properties as atoms.

\begin{lemma} \label{le:RelAtomProp1}
Let $(\Psi;\mathcal{E},\cdot,1,0)$ be a generalized information algebra. Then 
\begin{enumerate}
\item if $\alpha$ is an atom relative to $x$, then for all $\psi \in \Psi$ with $\epsilon_x(\psi) = \psi$, either $\alpha \cdot \psi = 0$ or $\alpha \cdot \psi = \alpha$,
\item if $\alpha$ is an atom relative to $x$, then for all $\psi \in \Psi$ either $\alpha \cdot \epsilon_x(\psi) = 0$ or $\epsilon_x(\psi)\ \leq \alpha$,
\item if $\alpha$ and $\beta$ are atoms relative to $x$, then either $\alpha = \beta$ or $\alpha \cdot \beta = 0$.
\end{enumerate}
\end{lemma}

This follows from Lemma \ref{le:AtomProp}. There are however further important properties of relative atoms, which are listed in the next Lemma. 

\begin{lemma} \label{le:RelAtomProp2}
Let $(\Psi;\mathcal{E},\cdot,1,0)$ be an atomic information algebra. Then,
\begin{enumerate}
\item if $\alpha$ is an atom, then, for all $x \in D$, $\epsilon_x(\alpha)$ is an atom relative to $x$,
\item if $\alpha$ is an atom with support $x$, then it is an atom relative to $x$.
\item for all $x \in D$ and atoms $\alpha'$ relative to $x$, there is an atom $\alpha$ such that $\alpha' = \epsilon_x(\alpha)$,
\item if $x \leq y$ and $\alpha$ is an atom relative to $x$, then there is an atom $\beta$ relative to $y$ such that $\alpha = \epsilon_x(\beta)$,
\item if $x \leq y$, $\alpha$ an atom relative to $x$, $\alpha \geq \epsilon_x(\psi)$ for $\psi \in \Psi$ with support $y$, then there exists an atom $\beta$ relative to $y$ such that $\beta \geq \psi$ and $\epsilon_x(\beta) = \alpha$.
\end{enumerate}
\end{lemma}

\begin{proof}
1.) If $\alpha$ is an atom, then $\alpha \not= 0$, hence $\epsilon_x(\alpha) \not= 0$. Consider now an element $\psi$ with support $x$, $\epsilon_x(\psi) = \psi$. We then have $\epsilon_x(\alpha \cdot \psi) = \epsilon_x(\alpha) \cdot \psi$. But, since $\alpha$ is an atom either $\alpha \cdot \psi = 0$ or $\alpha \cdot \psi = \alpha$. In the second case it follows that $\epsilon_x(\alpha) \cdot \psi = \epsilon_x(\alpha)$, in the first case $\epsilon_x(\alpha) \cdot \psi = 0$. So, $\epsilon_x(\alpha)$ is indeed an atom relative to $x$.

2.) is an immediate consequence of 1.)

3.) Since the algebra is atomic, there is an atom $\alpha \in At(\alpha')$, that is $\alpha' \leq \alpha$. It follows that $\alpha' = \epsilon_x(\alpha') \leq \epsilon_x(\alpha)$. But $\epsilon_x(\alpha) \not= 0$, therefore it follows that $\alpha' = \epsilon_x(\alpha)$.

4.) As before, there is an atom $\gamma \in At(\alpha)$ so that, as above, $\epsilon_x(\gamma) = \alpha$ and $\beta = \epsilon_y(\gamma)$ is an atom relative to $y$. But we have, since $x \leq y$, $\alpha = \epsilon_x(\gamma) = \epsilon_x(\epsilon_y(\gamma)) = \epsilon_x(\beta)$.

5.) We have $\alpha \cdot \psi \not= 0$ since $\epsilon_x(\alpha \cdot \psi) = \alpha \cdot \epsilon_x(\psi) = \alpha \not= 0$. So, there is an atom $\gamma \in At(\alpha \cdot \psi)$  such that $\gamma \geq \alpha \cdot \psi \geq \psi$, $\beta = \epsilon_y(\gamma)$ is an atom relative to $y$, and $\epsilon_x(\beta) = \epsilon_x(\epsilon_y(\gamma)) = \epsilon_x(\gamma) \geq \epsilon_x(\alpha \cdot \psi) = \alpha$. Since both $\alpha$ and $\epsilon_x(\beta)$ are atoms relative to $x$ it follows that $\epsilon_x(\beta) = \alpha$.
\end{proof}

Now, the following result shows that $At(\Psi)$ is a strongly order-generating set.

\begin{lemma} \label{le:StrongOrdGenAtoms}
Let $\alpha$ be an atom so that $\epsilon_x(\alpha) \geq \epsilon_x(\psi)$. Then there is an atom $\beta$ such that $\alpha \equiv_x \beta$ and $\beta \geq \psi$. 
\end{lemma}

\begin{proof}
We claim that $\epsilon_x(\alpha) \cdot \psi \not= 0$. In fact $\epsilon_x(\epsilon_x(\alpha) \cdot \psi) = \epsilon_x(\alpha) \cdot \epsilon_x(\psi) = \epsilon_x(\alpha) \not= 0$ Then there is an atom $\beta \in At(\epsilon_x(\alpha) \cdot \psi)$, so that $\beta \geq \epsilon_x(\alpha) \cdot \psi \geq \psi$ and $\epsilon_x(\beta) \geq \epsilon_x(\alpha)$. But since both $\epsilon_x(\beta)$ and $\epsilon_x(\alpha)$ are atoms relative to $x$ this implies that $\epsilon_x(\beta) = \epsilon_x(\alpha)$.
\end{proof}

So, according to Theorem \ref{th:GenRep_1}, the atomistic information algebra $(\Psi;\mathcal{E},\cdot,1,0)$ is embedded into the set algebra $(2^{At(\Psi)};g(\mathcal{E}),\cap,\Psi/\{0\},\emptyset)$. The embedding maps are given by
\begin{eqnarray*}
\psi \mapsto At(\psi) = \uparrow\!\psi \cap At(\Psi), \quad x \mapsto P_x,
\end{eqnarray*}
where $P_x$ is the partition of $At(\Psi)$ associated with the equivalence relation $\alpha \equiv_x \beta$ iff $\epsilon_x(\alpha) = \epsilon_x(\beta)$ in $At(\Psi)$. 

\begin{theorem}
An atomistic information algebra $(\Psi;\mathcal{E},\cdot,1,0)$ is embedded in the set algebra $(2^{At(\Psi)};g(\mathcal{E}),\cap,At(\Psi),\emptyset)$. 
\end{theorem}

A stronger statement is possible, if the information algebra $(\Psi;\mathcal{E},\cdot,1,0)$ is completely atomistic. Then the map $\psi \mapsto At(\psi)$ is surjective on $2^{At(\Psi)}$, and the pair of maps $(f,g)$ determine an isomorphism, hence the two algebras $(\Psi;\mathcal{E},\cdot,1,0)$ and $(2^{At(\Psi)};g(\mathcal{E}),\cap,At(\Psi),\emptyset)$ are isomorphic.

\begin{theorem} \label{th:ComplAtomInfAlg}
A completely atomistic information algebra $(\Psi;\mathcal{E},\cdot,1,0)$ is isomorphic to the set algebra $(2^{At(\Psi)};g(\mathcal{E}),\cap,At(\Psi),\emptyset)$. 
\end{theorem}

This representation theorem will be placed in the context of Boolean information algebras in the Section \ref{sec:BooleInfAlg} .

\subsubsection{Locally Atomic Information Algebras} \label{subsec:LocAtomInfAlg}

As we have seen above, the atoms of an atomic information algebra induce relative atoms. Many information algebras have no atoms, but may still have atoms relative to all domains $x$. This is an instance of the situation with locally order-generating sets. Given the concept of relative atoms, we denote the set of atoms relative to a domain $x$ by $At_x(\Psi)$. As a consequence of Lemma \ref{le:RelAtomProp1}, note that $\epsilon_x(At(\psi)) = At_x(\Psi)$. For an element $\psi$ of $\Psi$ we define the set
\begin{eqnarray*}
At_x(\psi) = \{\alpha \in At_x(\Psi):\epsilon_x(\psi) \leq \alpha\} = \uparrow\!\epsilon_x(\psi) \cap At_x(\Psi).
\end{eqnarray*}
of all atoms relative to $x$ which imply $\epsilon_x(\psi)$. 

The concepts of atomic, atomistic and completely atomistic information algebras can be extended to relative atoms. 

\begin{definition}
\textit{Locally Atomic and Atomistic Information Algebras:} Let $(\Psi;\mathcal{E},\cdot,1,0)$ be an information algebra. Then,
\begin{enumerate}
\item if for all $x \in D$, $\psi \not= 0$, the sets $At_x(\psi)$ are not empty, the algebra is called locally atomic,
\item if for  all $x \in D$, $\psi \not= 0$,
\begin{eqnarray*}
\epsilon_x(\psi) = \inf At_x(\psi)
\end{eqnarray*}
the algebra is called locally atomistic,
\item if the algebra is locally atomistic and if for all subsets $X$ of $At_x(\Psi)$, for all $x \in D$, there is an element $\psi \in \Psi$ such that
\begin{eqnarray*}
\psi = \epsilon_x(\psi) = \inf X,
\end{eqnarray*}
then the algebra is called locally completely atomistic.
\end{enumerate}
\end{definition}
We remark that any atomic, atomistic or completely atomistic information algebra is also locally atomic, locally atomistic or locally completely atomistic, but the converse does not hold. In a locally atomistic information algebra, the atoms relative to a domain $x$ may be considered to constitute the possible complete pieces of information relative to this domain, or, in other words, the possible answers to the question represented by $x$.

By Lemma \ref{le:RelAtomProp2}, the system $At_x(\Psi)$ for $x \in D$ is a tuple system of locally order-generating sets in the case of a locally atomistic information algebra. We may therefore apply the theory of Section \ref{subsec:LocOrdGenSet}. In the spirit of this section, we consider maps $a : D \rightarrow \bigcup_{x \in D} At_x(\Psi)$ so that $x \mapsto a_x \in At_x(\Psi)$. We restrict such maps to maps such that for $x \leq y$ we have $a_x = \epsilon_x(a_y)$. Let $U$ be the set of such consistent maps. Then, for every element $\psi$ in $\Psi$ we define
\begin{eqnarray*}
\mathbf{At}(\psi) = \{a \in U:a_x \in At_x(\psi), \forall x \in D\}.
\end{eqnarray*}
Then, according to Section \ref{subsec:LocOrdGenSet} the map $f : \Psi \rightarrow 2^U$ defined by
\begin{eqnarray*}
\psi \mapsto f(\psi) = \mathbf{At}(\psi)
\end{eqnarray*}
is a join homomorphism,
\begin{eqnarray*}
f(\phi \cdot \psi) &=& \mathbf{At}(\phi \cdot \psi) = \mathbf{At}(\phi) \cap \mathbf{At}(\psi) = f(\phi) \cap f(\psi), \\
f(0) &=& \emptyset, \\
f(1) &=& U.
\end{eqnarray*}
As in Section \ref{subsec:LocOrdGenSet}, $g$ maps $x \in D$ to the partition $P_x$ in $U$ defined by $a_x = b_x$. Then the pair of maps $(f,g)$ defined an embedding of the locally atomistic information algebra $(\Psi;\mathcal{E},\cdot,1,0)$ into the set algebra $(2^U;g(\mathcal{E}),\cap,U,\emptyset)$, see Theorem \ref{th:GenRep_2}.

\begin{theorem}
A locally atomistic information algebra $(\Psi;\mathcal{E},\cdot,1,0)$ is embedded into the set algebra $(2^U;g(\mathcal{E}),\cap,U,\emptyset)$. If it is locally completely atomistic, then it is isomorphic to this set algebra.
\end{theorem}

As we shall see in the following section, set algebras exhibit a Boolean structure, which information algebra do not have in general. But completely atomistic information algebra do inherit this structure, as the the last representation theorem indicates.

\subsection{Boolean Information Algebras} \label{sec:BooleInfAlg}

\subsubsection{Boolean Structure of Information}

Let $(\mathcal{S};\mathcal{E},\cap,U,\emptyset)$ be a set algebra. Then $\mathcal{S}$ is a join-subsemilattice of $(2^U;\subseteq)^{\vartheta}$. If $\mathcal{S}$ is a \textit{field} of sets, then the set algebra is called Boolean. Now, in an information algebra $(\Psi;\mathcal{E},\cdot,1,0)$, the semi lattice $(\Psi;\leq)$ may also be a Boolean algebra. The information algebra is then also called Boolean. But just as not any Boolean algebra is an algebra of subsets, there are Boolean information algebras, which are not set algebras. However by Stone duality every Boolean algebra is embedded into a field of sets. We show in this section that this extends to Boolean information algebras: Any Boolean information algebra is embedded into a Boolean set algebra.

First, we explain what a Boolean information algebra is.

\begin{definition} \label{def:BooleDomFree}
An information algebra $(\Psi;\mathcal{E},\cdot,1,0)$ is called Boolean, if $(\Psi;\leq)$ is a Boolean algebra in its information order.
\end{definition}

Examples of Boolean information algebras are quantifier algebras and cylindric algebras \cite{henkin71,kohlasschmid14,plotkin94}. Recall that in the information order $\phi \cdot \psi = \phi \vee \psi$, that the null element $0$ is the greatest element in this order and the unit $1$ the smallest one. So, we have $\phi \cdot \phi^{c} = 0$ and $\phi \wedge \phi^{c} = 1$. For further reference, we collect here a few results, well-known from Boolean algebras or monadic algebras (relating to existential quantifiers, see \cite{halmos62}):

\begin{lemma} \label{le:ExtrAndCompl}
Let $(\Psi;\mathcal{E},\cdot,1,0)$ be a Boolean information algebra. Then for $\phi,\psi \in \Psi$ and $x \in D$,
\begin{enumerate}
\item $\phi \leq \psi$ iff $\psi^{c} \leq \phi^{c}$,
\item $\phi \cdot \psi^{c} = 0$ iff $\psi \leq \phi$,
\item $\epsilon_x(\phi \wedge \psi) = \epsilon_x(\phi) \wedge \epsilon_x(\psi)$,
\item $\phi \wedge \epsilon_x(\phi) = \epsilon_x(\phi)$,
\item $\phi = \epsilon_x(\phi)$ iff $\phi^{c} = \epsilon_x(\phi^{c})$,.
\item $\phi = \epsilon_x(\phi)$ and $\psi = \epsilon_x(\psi)$ imply $\phi \wedge \psi = \epsilon_x(\phi \wedge \psi)$.
\end{enumerate}
\end{lemma}

\begin{proof}
Although, these results are known, they are may be not so easily accessible, since they refer to monadic Boolean algebras. Therefore the proof will be given here, except for the first two items which are classical Boolean algebra.

(3) Let $\eta = \phi \wedge \psi$. Then $\eta \leq \phi,\psi$ implies $\epsilon_x(\eta) \leq \epsilon_x(\phi),\epsilon_x(\psi)$. Hence $\epsilon_x(\eta)$ is a lower bound of $\epsilon_x(\phi)$ and $\epsilon_x(\psi)$. Let $\chi$ be another lower bound of $\epsilon_x(\phi)$ and $\epsilon_x(\psi)$. Then, by item 2 of the lemma, $\epsilon_x(\phi) \cdot \chi^{c} = 0$ and $\epsilon_x(\psi) \cdot \chi^{c} = 0$. It follows that
\begin{eqnarray}
0 = \epsilon_x(0) = \epsilon_x(\epsilon_x(\phi) \cdot \chi^{c}) = \epsilon_x(\phi) \cdot \epsilon_x(\chi^{c})
= \epsilon_x(\phi \cdot \epsilon_x(\chi^{c})).
\nonumber
\end{eqnarray}
This implies that $\phi \cdot \epsilon_x(\chi^{c}) = 0$. In the same way we obtain that $\psi \cdot \epsilon_x(\chi^{c}) = 0$. Using the distributive law in the Boolean algebra and the fact that combination is join, we obtain further
\begin{eqnarray}
0 &=& (\phi \cdot \epsilon_x(\chi^{c})) \wedge (\psi \cdot \epsilon_x(\chi^{c}))
= (\phi \wedge \psi) \cdot \epsilon_x(\chi^{c}) = \eta \cdot \epsilon_x(\chi^{c}).
\nonumber
\end{eqnarray}
From this it follows
\begin{eqnarray}
0 = \epsilon_x(0) = \epsilon_x(\eta \cdot \epsilon_x(\chi^{c})) = \epsilon_x(\eta) \cdot \epsilon_x(\chi^{c})
= \epsilon_x(\epsilon_x(\eta) \cdot \chi^{c}),
\nonumber
\end{eqnarray}
hence $\epsilon_x(\eta) \cdot \chi^{c} = 0$. But this implies that $\chi \leq \epsilon_x(\eta)$ and $\epsilon_x(\eta)$ is thus the greatest lower bound of $\epsilon_x(\phi)$ and $\epsilon_x(\psi)$. This proves that $\epsilon_x(\phi \wedge \psi) = \epsilon_x(\phi) \wedge \epsilon_x(\psi)$.

(4) This follows from $\epsilon_x(\phi) \leq \phi$.

(5) Assume $\phi = \epsilon_x(\phi)$. From $0 = \phi \cdot \phi^{c}$ we deduce $0 = \epsilon_x(\phi \cdot \phi^{c}) = \epsilon_x(\epsilon_x(\phi) \cdot \phi^{c}) = \epsilon_x(\phi) \cdot \epsilon_x(\phi^{c}) = \epsilon_x(\phi \cdot \epsilon_x(\phi^{c}))$. This implies that $\phi \cdot \epsilon_x(\phi^{c}) = 0$. On the other hand, we have also $1 = \phi \wedge \phi^{c}$, hence by item 3 proved above, $1 = \epsilon_x(\phi \wedge \phi^{c}) = \epsilon_x(\phi) \wedge \epsilon_x(\phi^{c}) = \phi \wedge \epsilon_x(\phi^{c})$. This shows that $\epsilon_x(\phi^{c})$ is the complement of $\phi$. By symmetry, the inverse implication follows too.

(6) This is a direct consequence of item 3 proved above.
\end{proof}

Boolean algebras have a dual algebra associated with the inverse order. This carries over to Boolean information algebras. Define the dual operations of combination and extraction as follows:
\begin{eqnarray}
\phi \cdot_{\vartheta} \psi &=& (\phi^{c} \cdot \psi^{c})^{c},
\nonumber \\
\epsilon_x^{\vartheta}(\phi) &=& (\epsilon_x(\phi^{c}))^{c}.
\nonumber
\end{eqnarray}
Let $\mathcal{E}^\vartheta$ be the set of operators $\epsilon^\vartheta_x$ for $x \in D$. It can easily be verified that $(\Psi;\mathcal{E}^\vartheta,\cdot_{\vartheta},0,1)$ is a Boolean information algebra, isomorphic to the original one under the maps $\phi \mapsto \phi^{c}$ and $\epsilon_x \mapsto \epsilon_x^{\vartheta}$. It is called the dual information algebra. The information-theoretic background of this duality will become clear in the case of Boolean set algebras, see the next Section \ref{subsec:FiniteBoole}.

The classical examples of Boolean information algebras are related to the algebras of algebraic logic, associated with propositional and predicate logic \cite{kohlas03}. These are monadic algebras, which are less general, and polyadic or cylindric algebras, which contain more operators than information algebras \cite{HMT71,halmos62,halmos98}. Quantifier algebras provide another example of Boolean information algebras, \cite{plotkin94}. Finally, set algebras
are Boolean information algebras.

Completely atomistic information algebras are also Boolean, isomorphic to the power set (information) algebra of $2^{At(\Psi)}$. In fact, the map $\psi \mapsto At(\psi)$ is a Boolean isomorphism: Assume $\phi$ maps to $At(\psi)^c = At(\phi)$. Then we have $At(\phi) \cap At(\psi) = \emptyset$, hence $\phi \cdot \psi = 0$. Let further $At(\phi) \cup At(\Psi)$ be the image of the element $\eta$ in $\Psi$. Then we have $\eta \leq \phi,\psi$. If $\chi$ is another lower bound of $\phi$ and $\psi$, then $At(\phi) \cup At(\psi) \subseteq At(\chi)$, hence $\chi \leq \eta$ and $\eta$ is thus the infimum of $\phi$ and $\psi$, $\eta = \phi \wedge \psi$. But $At(\phi) \cup At(\psi) = \Psi$ so that $\phi \wedge \psi = 1$. This shows that $\phi$ is the complement of $\psi$ and $\psi^c$ maps to $At(\psi)^c$. 

After this short introduction into Boolean information algebras, we turn to the question of representing them by set algebras. This question will be answered by extending known results from Stone duality theory for Boolean algebras. We start with the case of finite Boolean information algebras.

\subsubsection{Finite Boolean Algebras} \label{subsec:FiniteBoole}

Let $(\Psi;\mathcal{E},\cdot,1,0)$ be a Boolean information algebra. In this section we assume that $\Psi$ is \textit{finite}, hence a finite Boolean algebra. It is well known that finite Boolean algebras are power set algebras. We show here that this extends to finite Boolean information algebras. In fact, the approach is just as with Boolean algebras, with one exception: our information-theoretic  concept of \textit{atoms} corresponds to \textit{coatoms} in the order theoretic view. Although this changes nothing essential, for clarity's sake we shall present the ideas here in the framework of information algebras. We refer to \cite{davey90} for the representation theory of Boolean algebras.

First we note that in our terminology, $(\Psi;\mathcal{E},\cdot,1,0)$ is completely atomistic, which follows from classical Boolean algebra theory.

\begin{theorem} \label{th:AtomFiniteBoole}
A finite Boolean information algebra $(\Psi;\mathcal{E},\cdot,1,0)$ is completely atomistic.
\end{theorem}

As a corollary, it follows from Section \ref{subsec:AtomisticInfAlg}  that the information algebra $(\Psi;\mathcal{E},\cdot,1,0)$ is isomorphic (as an information algebra and as a Boolean lattice) to the set algebra $(2^{At(\Psi)};g(\mathcal{E}),\cap,At(\Psi),\emptyset)$, where the saturation operators $\sigma_{g(x)} = \sigma_{P_x}$ are associated with equivalence relations $\equiv_x$ for $x \in D$, that is, for any subset $S$ of $At(\Psi)$,
\begin{eqnarray}
\sigma_x(S) = \{\alpha \in At(\Psi):\exists \beta \in S \textrm{ such that}\ \epsilon_x(\alpha) = \epsilon_x(\beta)\}.
\nonumber
\end{eqnarray}
This means that
\begin{eqnarray}
At(\phi \cdot \psi) &=& At(\phi) \cap At(\psi),
\nonumber \\
At(0) &=& \emptyset,
\nonumber \\
At(1) &=& At(\Psi),
\nonumber \\
At(\epsilon_x(\phi)) &=&\sigma_x(At(\phi)).
\end{eqnarray}
But the map $\phi \mapsto A(\phi)$ is also an isomorphism of Boolean algebras That is, in addition, we have
\begin{eqnarray}
At(\phi \wedge \psi) &=& At(\phi) \cup At(\psi),
\nonumber \\
At(\phi^{c}) &=& (At(\phi))^{c}.
\end{eqnarray}

So, finite Boolean information algebras are set algebras with a finite universe:

\begin{theorem} \label{th:RepFiniteDomFreeBoole}
A finite Boolean information algebra $(\Psi;\mathcal{E},\cdot,1,0)$ is isomorphic as an information algebra as well as a Boolean lattice to the set algebra $(2^{At(\Psi)},g(\mathcal{E}),\cap,At(\Psi),\emptyset)$.
\end{theorem}

Note that in the representation of the information algebra as a set algebra of atoms, atoms represent possible worlds, and relative atoms relative to a domain $x$ represent precise answers to the question represented by $x$. So, the smaller a set $At_x(\psi)$, the more precise, the more informative is the information $\psi$. This is called the \textit{disjunctive} (or Sherlock Holmes) view: The unknown answer is this atom \textit{or} this one \textit{or} this one, \textit{or} ... etc. There is also the \textit{conjunctive} (or the collectors) view: The answer is this atom \textit{and} this one \textit{and} this one, \textit{and} ...etc. In this view $At(\psi)$ is the more informative, the larger the set. Note that this view is represented by the dual algebra of the set algebra $(2^{At(\Psi)};g(\mathcal{E}),\cap,At(\Psi),\emptyset)$. For more details on this dual view of information we refer to \cite{kohlasschneuwly09}.

As a preparation for subsequent sections let's remark that the principal ideals $\downarrow\!\alpha$ of atoms are maximal ideals in $\Psi$. In a finite algebra all ideals are principal. And we have
\begin{eqnarray*}
\downarrow\!\psi = \bigcap_{\alpha \in At(\psi)} \downarrow\!\alpha.
\end{eqnarray*}

The results for finite Booelan information algebras turn out to be generalizable to general nonfinite Boolean information algebras and distributive lattice information algebras..

\subsubsection{General Boolean  Algebras: Stone Duality Extended} \label{subsec:GenBooleInfAlg}

Here we start with a Boolean information algebra $(\Psi;\mathcal{E},\cdot,1,0)$. Such an algebra is in general not atomic and much less atomistic. But its ideal completion (as an information algebra) is so (Section \ref{subsec:IdCompl}). The atoms among ideals are the maximal ideals, which
are defined as follows.

\begin{definition} \label{def:MaxIdeal}
Let $(\Psi;\mathcal{E},\cdot,1,0)$ be an information algebra. An ideal $I \in I_{\Psi}$ of $\Psi$ is called maximal, if $I \not= \Psi$, $J \in I_{\Psi}$ and $I \subseteq J$ imply $I = J$ or $J = \Psi$.
\end{definition}

So maximal ideals are clearly atoms in the ideal completion $(I_{\Psi};\bar{\mathcal{E}},\cdot,\{1\},\Psi)$ of the information algebra $(\Psi;\mathcal{E},\cdot,1,0)$. The point is that the information algebra $(I_{\Psi},\bar{E};\cdot,\{1\},\Psi,\circ)$ is completely atomistic, if $(\Psi;\mathcal{E},\cdot,1,0)$ is Boolean. This follows from well known results of Boolean algebras (see below).

First let's give another definition relative to ideals.

\begin{definition} \label{def:PrimeIdeal}
Let $(\Psi;\mathcal{E},\cdot,1,0)$ be a Boolean information algebra. An ideal $I \in I_{\Psi}$ of $\Psi$ is called prime, if $I \not= \Psi$ and whenever $\phi \wedge \psi \in I$, then either $\phi\in I$ or $\psi \in I$.
\end{definition}

The existence of prime and maximal ideals is not a triviality. In fact, it needs some form of set theoretical existence theorem, such as the axiom of choice or variants of it \cite{davey90}. We take the existence of enough maximal respectively prime ideals for granted.  In  Boolean algebras prime and maximal ideals coincide. Further, for all $\psi \in \Psi$ there is a maximal ideal $I$ such that either $\psi \in I$ or $\psi^{c} \in I$. So, a maximal ideal $I$ is a consistent and complete theory in the sense that an ideal represents consistent information (see Section \ref{subsec:IdCompl}) and it is complete in the sense that each piece of information or its negation (complement) belongs to $I$. If $\phi \not= \psi$, then there is a maximal ideal which contains exactly one of the two elements. And finally for any ideal $J \not= \Psi$ in $I_{\Psi}$ there is a maximal ideal $I$ such that $J \subseteq I$, that is $J \leq I$ in the information order in the ideal completion. We refer to \cite{davey90} for these results about Boolean algebras. 

The point is now that the ideal completion $(I_{\Psi};\bar{\mathcal{E}};\cdot,\{1\},\Psi)$ is an \textit{atomic} information algebra since any ideal is contained in some maximal ideal, and it extends $(\Psi;\mathcal{E},\cdot,1,0)$. So we can apply the results of Section \ref{sec:AtomInfAlg}. Let $At(J)$ denote the set of atoms, that is, maximal ideals $I$ implying $J$, $J \leq I$. Further let $X(\Psi)$ denote the set of all maximal ideals of $\Psi$, or all atoms of $I_\Psi$. Then, the map $J \mapsto At(J)$ is a homomorphism from the information algebra $(I_{\Psi};\bar{\mathcal{E}},\cdot,\{1\},\Psi)$ into the set algebra of sets of atoms $(2^{X(\Psi)};g(\bar{\mathcal{E}}),\cap,X(\Psi),\emptyset)$, where $g(\bar{\mathcal{E}})$ is the set of all saturation operators $\sigma_{P_x}$ related to the partition $P_x$ of $X(\Psi)$ defined by the equivalence relation $\bar{\epsilon}_x(I) = \bar{\epsilon}_x(K)$,
\begin{eqnarray}
\sigma_{P_x}(A) = \{I \in X(\Psi): \exists K \in A \textrm{ such that}\ \bar{\epsilon}_x(I) = \bar{\epsilon}_x(K)\}
\nonumber
\end{eqnarray}
for any subset $A$ of $X(\Psi)$, see Section \ref{sec:AtomInfAlg}. 

As claimed above, the information algebra $(I_{\Psi};\bar{\mathcal{E}},\cdot,\{1\},\Psi)$ is in fact completely atomistic: Indeed, for any ideal $J$ in $I_{\Psi}$, we have that $J \subseteq \cap At(J)$. If we assume that $J \not= \cap At(J)$, then there must be an element $\psi$ in $\cap At(J)$ but not in $J$. If $\phi \in J$, then there is either a maximal ideal $I$ containing $\phi$ but not $\psi$ or one containing $\psi$ but not $\phi$. The latter case is excluded because this ideal could not belong to $At(J)$, in the former case we have $\psi \not\in \cap At(J)$ against the assumption, and we must therefore have $J = \cap At(J)$. This means that the ideal completion is atomistic. Further the intersection of any family of ideals, in particular of any family of maximal ideals, is still an ideal. So, the information algebra is completely atomistic. By Theorem \ref{th:ComplAtomInfAlg}, the information algebra $(I_{\Psi};\bar{\mathcal{E}},\cdot,\{1\},\Psi)$ is isomorphic to the set algebra $(2^{X(\Psi)};g(\bar{\mathcal{E}}),\cap,X(\Psi),\emptyset)$ via the pair of maps $J \mapsto At(J) = \uparrow\!J \cap X(\Psi)$ and $\bar{\epsilon}_x \mapsto \sigma_{P_x}$.

The information algebra $(\Psi;\mathcal{E},\cdot,1,0)$ is embedded into the information algebra $(I_{\Psi};\bar{\mathcal{E}},\cdot,\{1\},\Psi)$ by the pair of maps $\psi \mapsto \downarrow\!\psi$ and $\epsilon_x \mapsto \bar{\epsilon}_x$. Further, a principal ideal $\downarrow\!\psi$ maps to the set $At(\downarrow\!\psi)$ of maximal ideals containing it. Let's denote this set of maximal ideals by $X_{\psi}$, that is
\begin{eqnarray}
X_{\psi} = \{I \in X(\Psi):\psi \in I\}.
\nonumber
\end{eqnarray}
We have also $X_\psi = \uparrow\!\psi \cap X(\Psi)$, where here $\uparrow\!\psi = \{I \in I_\Psi:\psi \in I\}$ is the upset of all elements (ideals) in $I_\Psi$ greater than $\psi$. The pair of composed maps $f : \psi \mapsto X_{\psi}$ and $g : x \mapsto P_x$ is an information algebra homomorphism from the information algebra $(\Psi;\mathcal{E},\cdot,1.0)$ into the information algebra $(2^{X(\Psi)};g(\bar{\mathcal{E})},\cap,At(X(\Psi)),\emptyset)$. Here it is understood that $g(\epsilon_x) = \sigma_{P_x}$, where $\sigma_{P_x}$ is the saturation operator in $X(\Psi)$ associated with the partition $P_x$. This means that
\begin{eqnarray} \label{eq:CompHom}
X_{\phi \cdot \psi} &=& X_{\phi} \cap X_{\psi}, \\
X_{1} &=& X(\Psi), \nonumber \\
X_{0} &=& \emptyset,\nonumber  \\
X_{\epsilon_x(\psi)} &=& \sigma_{P_x}(X_{\psi}) \nonumber .
\end{eqnarray}
In fact, this map is an embedding, and this gives us a first version of a representation theorem for a Boolean information algebras. We shall see below that it is even an embedding of a Boolean algebra.

\begin{theorem} \label{th:DomFreeBooleRep}
A Boolean information algebra $(\Psi;\mathcal{E},\cdot,1,0)$ is embedded into the set algebra $(2^{X(\Psi)};g(\bar{\mathcal{E}}),\cap,X(\Psi),\emptyset)$ by the pair of maps $\phi \mapsto X_{\phi}$ and $\epsilon \mapsto \sigma_{\epsilon}$.
\end{theorem}

\begin{proof}
We noted already that the pair of maps is a homomorphism. It remains to show that it is one-to-one. In fact, if $\phi \not= \psi$, then there is a maximal ideal $I$ which contains one, but not the other element. So $X_{\phi} \not= X_{\psi}$.
\end{proof}

Stone duality for Boolean algebra allows to describe this representation of a Boolean information algebra more precisely. This means to characterize explicitly the subsets $X_\phi$ within $X(\Psi)$. This is accomplished by introducing an appropriate topology on the set $X(\Psi)$. Actually 
\begin{eqnarray}
\mathcal{B} =\{X_{\phi}:\phi \in \Psi\}
\nonumber
\end{eqnarray}
is an open base by (\ref{eq:CompHom}) (we refer to \cite{davey90} for this and all other issues regarding Stone duality). We should remark that in the literature usually $X_{\phi}$ is defined as the set of maximal ideals $\{I \in X(\Psi):\phi \not\in I\}$. But if $\phi \not\in I$, then $\phi^{c} \in I$ such that this set is $X_{\phi^{c}}$ in our terminology. So, this changes nothing essential, but for our purpose the present definition of $X_{\phi}$ is more natural and appropriate. Note that as a consequence of Theorem \ref{th:DomFreeBooleRep} the system $(\mathcal{B};g(\bar{\mathcal{E}}),\cap,X(\Psi),\emptyset)$ is an information algebra, a subalgebra of $(2^{X(\Psi)};g(\bar{\mathcal{E}}),\cap,X(\Psi),\emptyset)$. Of course, we consider here the restriction of the saturation operators $\sigma_{P_x}$ to $\mathcal{B}$.

The open sets of the topology $\mathcal{T}$ are unions of base sets,
\begin{eqnarray}
\mathcal{T} = \{U \subseteq X(\Psi):U \textrm{ is a union of members of}\ \mathcal{B}\}.
\nonumber
\end{eqnarray}
The topological space $(X(\Psi),\mathcal{T})$ is called the \textit{dual} or \textit{prime ideal space} of $\Psi$. Since $X_{\phi^{
c}} = X^{c}_{\phi}$, the open sets $X_{\phi}$ are also closed. In fact, $\mathcal{B}$ coincides exactly with the \textit{clopen} subsets of $X(\Psi)$. These clopen sets form not only an information algebra, but also a Boolean algebra, hence a Boolean information algebra. The space $X(\Psi)$ is \textit{compact}. Further, if $I \not= J$ are two maximal ideals, then there exists a clopen set $X_{\phi}$ such that $I \in X_{\phi}$ and $J \not\in X_{\phi}$. That is, the topological space $(X(\Psi),\mathcal{T})$ is a $T_{0}$ space; actually it is totally disconnected, that is, if $I \not= J$, then there exist disjoint clopen sets $X_\phi$ and $X_\psi$ such that 
$I \in X_\phi$ and $J \in X_\psi$. A compact, totally disconnected  topological space is also called a \textit{Boolean space}. The Stone representation theorem asserts that the map $\phi \mapsto X_{\phi}$ is a \textit{Boolean isomorphism} of $\Psi$ onto the Boolean algebra of clopen sets of the dual space $(X(\Psi),\mathcal{T})$. Since our definition of $X_{\phi}$ is the same as the usual definition of $X_{\phi^{c}}$, this means that we have to take the inverse order in the power set of $X(\Psi)$ such that
\begin{eqnarray}
X_{\phi \vee \psi} &=&X_{\phi} \cap X_{\psi} \textrm{ (this is the combination operation)},
\nonumber \\
X_{\phi \wedge\psi} &=&X_{\phi} \cup X_{\psi},
\nonumber \\
X_{\phi^{c}} &=& X^{c}_{\phi},
\nonumber \\
X_{1} &=& X(\Psi),
\nonumber \\
X_{0} &=& \emptyset,
\nonumber \\
X_{\epsilon(\phi)} &=& \sigma_{\epsilon}(X_{\phi}).
\end{eqnarray}

This allows to extend the representation theorem Theorem \ref{th:DomFreeBooleRep} for a Boolean information algebra as follows, since the map $\phi \mapsto X_{\phi}$ maps $\Psi$ to the family of clopen sets $\mathcal{B}$ of $X(\Psi)$

\begin{theorem} \label{th:DomFreeBooleRepExt}
A Boolean information algebra $(\Psi;\mathcal{E},\cdot,1,0,\circ)$ is isomorphic to the set algebra $(\mathcal{B};g(\bar{\mathcal{E}}),\cap,X(\Psi),\emptyset)$, where $g(\bar{\mathcal{E}})$ is the set of saturation operator $\sigma_{P_x}$ for $x \in D$, by the pair of maps $\phi \mapsto X_{\phi}$ and $x \mapsto P_x$, both as an information algebra as well as a Boolean algebra.
\end{theorem}

So, Boolean information algebras (including finite ones)  have a very satisfactory information-theoretic representation: Any piece of information is represented by the set of consistent and complete theories it is contained in. Combining two pieces of information consists in selecting the theories containing both of the pieces. Extraction of information from a piece $\phi$ means to collect all consistent and complete theories which contain the piece of information extracted from $\phi$. And that is exactly the collection of all theories which contain some piece of information $\psi$ whose extracted information equals the one of $\phi$, since
\begin{eqnarray}
X_{\epsilon_x(\phi)} = \sigma_{P_x}(X_{\phi}) = \{I \in X(\Psi): \exists \psi \in I \textrm{ such that}\ \epsilon(\psi) = \epsilon(\phi)\}.
\nonumber
\end{eqnarray}

There exists a well-known duality theory between Boolean algebras and Boolean spaces. This theory could be extended to Boolean information algebras. This has been done for the case of commutative Boolean algebras in \cite{kohlasschmid16}. This could be extended to the present more general case, but we renounce to work it out here.

\subsection{Distributive Lattice Information Algebra} \label{subsec:DistLattInfAlg}

\subsubsection{Distributive Lattice Structure of Information} \label{subsec:DistLattStruct}

Boolean algebras are distributive lattices. In the same way, in a Boolean information algebra $(\Psi;\mathcal{E},\cdot,1,0)$, the Boolean algebra $\Psi$ is a bounded distributive lattice. Therefore, we consider in this section information algebras where $\Psi$ is a bounded distributive lattice. Distributive lattices have a well developed representation and duality theory, the so-called Priestley-duality theory, generalizing the one of Boolean algebras \cite{davey90}, which we take here as a base to develop a corresponding representation theory of distributive lattice information algebras. Further, \cite{cignoli91} studied existential quantifiers on distributive lattices and his results are exactly what is needed to extend representation theory of distributive lattices to information algebras where $\Psi$ is a distributive lattice.

In the case of Boolean information algebras it was sufficient to assume $\Psi$ to be a Boolean algebra. Then it follows for instance that
\begin{eqnarray} \label{eq:DistQuantOverMeet}
\epsilon_x(\phi \wedge \psi) = \epsilon_x(\phi) \wedge \epsilon_x(\psi).
\end{eqnarray}
In case that $\Psi$ is a general distributive lattice this identity can no more be derived  \footnote{Take the Boolean algebra $\{\bot,t,f,\top\}$ and add a further element $0$ below $\bot$. Define on this distributive lattice $\epsilon(x) = x$ for $x \not=\bot,0$ and $\epsilon(x) = 0$ otherwise. Then $\epsilon$ satisfies all conditions of an existential quantifier, but not (\ref{eq:DistQuantOverMeet}).}. But it nevertheless is an essential element of a theory of distributive lattice information algebras, especially in relation to the concept of an existential quantifier on a distributive lattice. So it has to be  explicitly required in the definition of a distributive information algebra.

\begin{definition} \label{def:DistLattInfAlg}
A domain-free information algebra $(\Psi;\mathcal{E},\cdot,1,0)$ is called a distributive lattice information algebra, if
\begin{enumerate}
\item $(\Psi;\leq)$ is a distributive lattice,
\item For all $x \in D$ and $\phi,\psi \in \Psi$ (\ref{eq:DistQuantOverMeet}) holds.
\end{enumerate}
\end{definition}

An example of a distributive lattice information algebra is a lattice-valued algebra, see Section \ref{subsec:LattValAlg}.

\subsubsection{The Finite Case: Birkhoff Duality Extended}

In this subsection, we assume that $\Psi$ is \textit{finite} and derive a representation theorem for this case. Of course, the resulting theory will be a special case of the general theory and could be derived from it. But in many respects it is simpler than the general theory and thus worthwhile to develop it in its own right. In \cite{davey90} the finite case is also presented before the general case. We claim that in this case the order-generating set is given by the \textit{meet-irreducible} elements of the finite lattice $(\Psi;\leq)$.

\begin{definition}
\textit{Meet Irreducible Elements in a Lattice:} An element $\chi \not= 0$ of a lattice $(\Psi;\leq)$ is called meet irreducible, if $\chi = \phi \wedge \psi$ implies either $\chi = \phi$ or $\chi = \psi$.
\end{definition}

Here follow a few equivalent definitions of meet-irreducibility in a distributive lattice:

\begin{lemma} \label{le:MeetIrred}
Let $(\Psi;\leq)$ be a distributive lattice, and $\psi \in\Psi$, not equal to $0$. Then the following are equivalent
\begin{enumerate}
\item $\phi$ is meet-irreducible,
\item for all $\chi,\eta \in \Psi$, $\phi \geq \chi \wedge \eta$ implies $\phi \geq \chi$ or $\phi \geq \eta$,
\item for all $n=1,2,\ldots$; $\phi_{1},\ldots,\phi_{n} \in \Psi$ and $\phi \geq \phi_{1} \wedge \ldots \wedge \phi_{n}$ implies $\phi \geq \phi_{i}$ for some $i \in \{1,\ldots,n\}$.
\end{enumerate}
\end{lemma}

The proof of this, as well as the following lemma, can be found in \cite{davey90}, were the proofs are given in the dual version for join-irreducible elements. Note that atoms are meet-irreducible elements and in a Boolean algebra they are the only meet-irreducible elements. In so far is the theory which follows a generalization of the case of finite Boolean information algebras. Denote the meet-irreducible elements of $\Psi$ by $\mathcal{M}(\Psi)$. The following theorem implies that this set is order-generating.

\begin{theorem} \label{th:PrimeOrdGen}
If $\phi \not\leq \psi$ for $\phi,\psi$ elements of the finite distributive lattice $(\Psi;\leq)$, then there is a meet-irreducible element $\chi$ in $\mathcal{M}(\Psi)$ such that $\psi \leq \chi$ and $\phi \not\leq \chi$.
\end{theorem}

\begin{proof}
Define $X = \{\eta \in \Psi:\psi \leq \eta,\phi \not\leq \eta\}$. Since $\Psi$, hence $X$ are finite sets, there is a maximal element $\chi$ in $X$. We claim that $\chi$ is meet-irreducible. Suppose $\chi = \mu \wedge \lambda$ and $\chi < \mu$, $\chi < \lambda$. By the maximality of $\chi$ neither $\mu$ nor $\lambda$ can belong to $X$. We have $\psi \leq \chi < \mu$ and $\psi \leq \chi < \lambda$, so that $\psi \leq \mu,\lambda$. Therefore, since $\mu,\lambda \not\in X$, this implies $\phi \leq \mu,\lambda$. But then $\chi = \mu \wedge \lambda \geq \phi$ which is a contradiction, since $\chi \in X$. So we must have either $\chi = \mu$ or $\chi = \lambda$.
\end{proof}

From Theorem \ref{th:CharOrdGen} it follows that $\mathcal{M}(\Psi)$ is indeed an order-generating set. So, we have for all elements $\psi$ of $\Psi$
\begin{eqnarray*}
\psi = \inf \uparrow\!\psi \cap \mathcal{M}(\Psi).
\end{eqnarray*}
Further, the map $\psi \mapsto X_\psi = \uparrow\!\psi \cap \mathcal{M}(\Psi) = \{\chi \in \mathcal{M}(\Psi):\psi \leq \chi\}$ is a lattice homomorphism of the finite distributive lattice $(\Psi;\leq)$ into the field $(2^{\mathcal{M}(\Psi)};\subseteq)^\vartheta$ such that

\begin{eqnarray*}
X_{\phi \vee \psi} = X_{\phi \cdot \psi} &=& X_\phi \cap X_\psi, \\
X_{\phi \wedge \psi} &=& X_\phi \cup X_\psi, \\
X_1 &=& \mathcal{M}(\Psi), \\
X_0 &=& \emptyset.
\end{eqnarray*}

The map $X_\phi$ is injective, hence a lattice isomorphism between the finite lattice $(\Psi;\leq)$ and the lattice  $(\mathcal{U}(\mathcal{M}(\Psi));\subseteq)^\vartheta$ of up-sets of meet-irreducible elements of $\Psi$. In order to extend this isomorphism to the finite distributive lattice information algebra $(\Psi;\mathcal{E},\cdot,1,0)$, we must extend the map to extraction operators. Or, in other words, we must show that $\mathcal{M}(\Psi)$ is a strongly order-generating set (see Section \ref{subsec:OrdGenSets}). This follows from the following two lemmas. Remind that $\phi \equiv_x \psi$ iff $\epsilon_x(\phi) = \epsilon_x(\psi)$.

\begin{lemma} \label{le:FiniteDistrAux1}
Let $(\Psi;\mathcal{E},\cdot,1,0)$ be a finite distributive lattice information algebra and $x \in D$. Then, if $\eta$ and $\chi$ are meet-irreducible elements of $\Psi$ such that $\epsilon_x(\eta) \leq \epsilon_x(\chi)$, there exists a meet-irreducible element $\lambda$ of $\Psi$ such that $\lambda \equiv_x \chi$ and $\eta \leq \lambda$.
\end{lemma}

\begin{proof}
Consider $\eta \cdot \epsilon_x(\chi)$ and the set $F = \{\psi:\psi = \epsilon_x(\psi) \not\leq\chi\}$. The set $F$ is finite, hence $\wedge F$ exists. Further, by (\ref{eq:DistQuantOverMeet}) we have $\wedge F = \wedge_{\psi \in F}\ \epsilon_x(\psi) = \epsilon_x(\wedge F)$. It follows from Lemma \ref{le:MeetIrred} (3) that $F = \wedge F \not\leq\chi$, hence $\wedge F \in F$. Assume that $\eta \cdot \epsilon_x(\chi) \geq \wedge F$. Then it follows from $\epsilon_x(\eta) \leq \epsilon_x(\chi)$ that $\epsilon_x(\eta \cdot \epsilon_x(\chi)) = \epsilon_x(\eta) \cdot \epsilon_x(\chi) = \epsilon_x(\chi) \geq \wedge F$. But then $\chi \geq \epsilon_x(\chi) \geq \wedge F$ which contradicts $\wedge F \in F$. Therefore we conclude that $\wedge F \not\leq \eta \cdot \epsilon_x(\chi)$. By Theorem \ref{th:PrimeOrdGen} there is then a meet-irreducible element $\lambda$ in $\Psi$  such that $\eta \cdot \epsilon_x(\chi) \leq \lambda$ and $\wedge F \not\leq \lambda$. Then $\eta \leq \eta \cdot \epsilon_x(\chi) \leq \lambda$. Further, we have $\epsilon_x(\eta \cdot \epsilon_x(\chi)) = \epsilon_x(\chi) \leq \epsilon_x(\lambda)$. But on the other hand we have also $\epsilon_x(\lambda) \leq \lambda$ and $\epsilon_x(\lambda) \not\in F$, since $\wedge F \leq \epsilon_x(\lambda) \leq \lambda$ would contradict $\wedge F \not\leq \lambda$. Therefore, since $\epsilon_x(\epsilon_x(\lambda)) = \epsilon_x(\lambda)$, this means that $\epsilon_x(\lambda) \leq \chi$, hence $\epsilon_x(\lambda) \leq \epsilon_x(\chi)$. Thus finally, we have $\epsilon_x(\lambda) = \epsilon_x(\chi)$.
\end{proof}

\begin{lemma} \label{le:FiniteDistrAux2}
Let $(\Psi;\mathcal{E},\cdot,1,0)$ be a finite distributive lattice information algebra, $x \in D$, $\chi$ a meet-irreducible element of $\Psi$ and $\phi \in \Psi$ such that $\epsilon_x(\phi) \leq \epsilon_x(\chi)$. Then there is a meet-irreducible element $\lambda$ of $\Psi$ such that $\phi \leq \lambda$ and $\lambda \equiv_x \chi$.
\end{lemma}

\begin{proof}
If $\phi \leq \chi$, then take $\lambda = \chi$. So, assume that $\phi \not\leq \chi$ and consider the set $F = \{\psi:\psi = \epsilon_x(\psi) \not\leq \chi\}$. As in the proof of the preceding lemma we conclude that $\wedge F \in F$. Assume that $\wedge F \leq \phi$. But then $\wedge F = \epsilon_x(\wedge F) \leq \epsilon_x(\phi) \leq \epsilon_x(\chi) \leq \chi$. But this contradicts $\wedge F \in F$. Hence we have $\wedge F \not\leq \phi$. Then, by Theorem \ref{th:PrimeOrdGen}, there is a meet-irreducible element $\eta$ such that $\phi \leq \eta$ and $\wedge F \not\leq \eta$. Now, $\psi = \epsilon_x(\psi) \not\leq \chi$ is equivalent to $\psi = \epsilon_x(\psi) \not\leq \epsilon_x(\chi)$. By Lemma \ref{le:MeetIrred} (3) it follows from $\wedge F \not\leq \eta$ that $\psi \in F$ implies $\psi \not\leq \eta$. Thus $\epsilon_x(\eta) \not= \psi$ for all $\psi \in F$. It follows that $\epsilon_x(\eta) \leq \chi$, hence $\epsilon_x(\eta) \leq \epsilon_x(\chi)$.

Apply now Lemma \ref{le:FiniteDistrAux1} to obtain a meet-irreducible element $\lambda$ such that $\eta \leq \lambda$ and $\epsilon_x(\lambda) = \epsilon_x(\chi)$. Then $\phi \leq \eta$ implies $\phi \leq \lambda$ and the lemma is proved.
\end{proof}
This second lemma shows that $\mathcal{M}(\Psi)$ is indeed strongly order-generating.

This allows to formulate the main theorem about finite distributive lattice information algebras. Let $\sigma_{\epsilon}$ be the saturation operator associated with this equivalence relation $\equiv_x$,
\begin{eqnarray}
\sigma_{\epsilon}(S) = \{\chi \in \mathcal{M}(\Psi): \exists \eta \in S \textrm{ such that}\ \chi \equiv_{\epsilon} \eta\}
\end{eqnarray}
defined for any subset $S$ of $\mathcal{M}(\Psi)$. As usual, $g$ maps $x \in D$ to $P_x$, where $P_x$ is the partition in $\mathcal{M}(\Psi)$ associated with the equivalence relation $\equiv_x$.

\begin{theorem} \label{th:RepFiniteDomFreeBoole}
A finite distributive lattice information algebra $(\Psi;\mathcal{E},\cdot,1,0)$ is embedded as an information algebra as well as a lattice in the set-salgebra\\ $(2^{\mathcal{M}(\Psi)};g(\mathcal{E}),\cap,\mathcal{M}(\Psi),\emptyset)$.
\end{theorem}

As part of this theorem, we have the result that
\begin{eqnarray*}
X_{\epsilon_x}(\psi) = \sigma_{P_x}(X_\psi).
\end{eqnarray*}
This shows that $\sigma_{P_x}$ maps  $\mathcal{U}(\mathcal{M}(\Psi))$ into itself. This implies that the distributive lattice information algebra $(\Psi,\mathcal{E},\cdot,1,0)$ is isomorphic to the subalgebra $(\mathcal{U}(\mathcal{M}(\Psi));g(\mathcal{E});\\\cap,\mathcal{M}(\Psi),\emptyset)$ of upsets of meet-irreducible elements of $\Psi$. So, distributive lattice information algebras are isomorphic, that is essentially identical, to information algebras of up-sets of ordered sets with intersection as combination.  This is the starting point of a duality theory, which we do not pursue here.

There are information algebras $(\Psi;\mathcal{E},\cdot,1,0)$ where $(\Psi;\leq)$ is a distributive lattice, but (\ref{eq:DistQuantOverMeet}) does not hold or is a non-distributive lattice. Then the theory developed above does not apply. On the other hand, any information algebra is embedded into a set algebra, which is distributive and where saturation operators do satisfy (\ref{eq:DistQuantOverMeet}). The point is of course that this embedding is \textit{not} a lattice homomorphism, as it is in the present case.

\subsubsection{The General Case: Priestley Duality Extended} \label{subsec:GenDistrLattInfAlg}

We now turn to the general case of a distributive lattice information algebra $(\Psi;\mathcal{E},\cdot,1,0)$. Then $(\Psi;\leq)$ is a distributive lattice and we may apply the well-known Priestley duality theory to $\Psi$ to represent it by a subset lattice. The key notion in this theory is the one of a prime ideal, see Definition \ref{def:PrimeIdeal}. Let $X(\Psi)$ denote the set of prime ideals of the distributive lattice $\Psi$. The Priestley representation theorem says that the map $\psi \mapsto \{I \in X(\Psi):\psi \not\in I\}$ is a lattice embedding of $\Psi$ into the the power set of $X(\Psi)$, \cite{davey90}. As in the section on Boolean information algebras we have to consider that the information order is the opposite of the inclusion order and for this reason we consider the map $\psi \mapsto \{I \in X(\Psi):\psi \in I\}$. However, in this context $X(\Psi)$ is not an order-generating subset of the ideal completion $I_\Psi$ of the distributive lattice information algebra as one might expect. In fact, by Theorem I-3.25 in \cite{gierz03}, if the prime element of $I_\Psi$, that is the prime ideals of $\Psi$, were order-generating in $I_\Psi$, then $I_\Psi$ must be a frame, which in general is not the case. Nevertheless  Priestley duality theory extends to distributive lattice information algebras.                                                     

Define 
\begin{eqnarray} \label{eq:PriestleyRTepMap}
X_{\psi} = \{I \in X(\Psi):\psi \in I\}.
\nonumber
\end{eqnarray}
It follows that
\begin{eqnarray}
X_{\phi \cdot \psi} = X_{\phi \vee \psi} &=& X_{\phi} \cap X_{\psi},
\nonumber \\
X_{\phi \wedge \psi} &=& X_{\phi} \cup X_{\psi}.
\nonumber
\end{eqnarray}
Further we have $X_{0} = \emptyset$ and $X_{1} = X(\Psi)$. This shows that the map $\psi \mapsto X_{\psi}$ is a lattice homomorphism from $\Psi$ into $(2^{X(\Psi)};\subseteq)^{\vartheta}$. Further, the map is one-to-one, hence an embedding of the distributive lattice into $(2^{X(\Psi)},\subseteq)^{\vartheta}$. In fact, in a bounded distributive lattice, if $\phi \not\leq \psi$, then there exists a prime ideal $P$ sucht that $\phi \not\in P$ but $\psi \in P$ \cite{davey90}. So we have the desired map $X_\psi$ with the small difference that we map $\phi$ to the set of prime ideals containing it, instead of the usual map to the set of prime ideals not containing it. This makes sense from the information-theoretic point of view. Prime ideals are consistent complete theories or collections of information elements: As ideals they are \textit{consistent} in the sense that they contain with any element all elements implied by it and with any two elements also their combination. They are \textit{complete} theories, in the sense that if they contain $\phi \wedge \psi$,  they contain $\phi$ or $\psi$. And the map $X_{\psi}$ assigns to an element $\psi$ all consistent and complete theories $I$ which are consistent with $\psi$, that is, contain $\psi$. Therefore, in the sequel we are going to express Priestley duality theory for distributive lattices in this view.

The family of sets
\begin{eqnarray}
\mathcal{B} = \{X_{\phi} \cap (X(\Psi) - X_{\psi}):\phi,\psi \in \Psi\}
\nonumber
\end{eqnarray}
is the basis of a topological space $(X(\Psi),\mathcal{T}(\Psi))$. This topological space is \textit{compact} and it is ordered by inclusion. Further, the \textit{clopen} (simultaneously open and closed) subsets are finite unions of the form $X_{\phi} \cap (X(\Psi) - X_{\psi})$ for $\phi,\psi \in \Psi$. Since $X_{1} = X(\Psi)$ belongs to $\mathcal{B}$, the sets $X_{\psi}$ are exactly the \textit{clopen up-sets} of this ordered topological space. We denote the clopen up-sets in $(X(\Psi),\mathcal{T}(\Psi))$ by $\mathcal{U}(X(\Psi))$. So, $\Psi$ is lattice-isomorphic to the lattice of clopen upsets of the topological space $(X(\Psi),\mathcal{T}(\Psi))$, more precisely to the lattice $(\mathcal{U}(X(\Psi));\subseteq)^\vartheta$ with the order inverse to inclusion: $\phi \leq \psi$ implies $X(\phi) \supseteq X_{\psi}$. Finally it holds also that if $P,Q \in X(\Psi)$ such that $P \not\subseteq Q$, then there is a clopen up-set $U$ such that $P \in U$ and $Q \not\in U$. This means that the ordered topological space $(X(\Psi),\mathcal{T}(\Psi))$ is \textit{totally order-disconnected}. A compact, totally order-disconnected topological space is called a \textit{Priestley space}, see \cite{davey90}. So, $(X(\Psi),\mathcal{T}(\Psi))$ is a Priestley space.

In order to extend this representation theory to distributive lattice information algebras, we need to extend the theory to include extraction. For this part, we use the work of \cite{cignoli91} on distributive lattices with an existential quantifier. In presenting his results, the reader should again be aware, that \cite{cignoli91} works in the lattice $(\Psi;\leq)^\vartheta$, where the order is inverse to our natural information order. We present however here Cignoli's results with respect to the lattice $\Psi$ with our usual information order.

We define $\epsilon_x(I) = I \cap \epsilon_x(\Psi)$ for any prime ideal $I \in \epsilon_x(\Psi)$. This is not to be confounded with the extraction operator $\bar{\epsilon}_x(I)$ in the ideal completion, see Section \ref{subsec:IdCompl}. We remark however, that due to Lemma \ref{le:ExtrOfideals} $\epsilon_x(I) = \epsilon_x(J)$ is equivalent to the equivalence relation $I \equiv_{\bar{\epsilon_x}} J$ in the information algebra of the ideal completion. Note that $\epsilon_x(I)$ is a prime ideal in $\epsilon_x(\Psi)$. The following results which generalize Lemmas \ref{le:FiniteDistrAux1} and \ref{le:FiniteDistrAux2}, are adapted from \cite{cignoli91}.

\begin{lemma} \label{leCignoli1}
Given $P,Q \in X(\Psi)$ such that $\epsilon_x(Q) \subseteq \epsilon_x(P)$, there is a $R \in X(\Psi)$ such that $\epsilon_x(R) = \epsilon_x(P)$ and $Q \subseteq R$.
\end{lemma}

\begin{proof}
Let $I$ be the ideal in $\Psi$ generated by $Q \cup \epsilon_x(P)$ and $F$ the filter in $\Psi$ generated by $\epsilon_x(\Psi) - P$. Assume $I \cap F \not= \emptyset$ and consider $\phi \in I \cap F$. Then there must be $\eta \in \epsilon_x(\Psi) - P$ such that $\eta = \epsilon_x(\eta) \leq \phi$ and a $\psi \in Q$, a $\chi \in \epsilon_x(P)$ so that $\phi \leq \psi \vee \chi = \psi \vee \epsilon_x(\chi)$. Then we obtain $\eta = \epsilon_x(\eta) \leq \phi \leq \psi \vee \epsilon_x(\chi)$ so that $\eta = \epsilon_x(\eta) \leq \epsilon_x(\psi \vee \epsilon_x(\psi)) = \epsilon_x(\psi) \vee \epsilon_x(\chi)$. But $\psi \in Q$ implies $\epsilon_x(\psi) \in \epsilon_x(Q) \subseteq \epsilon_x(P)$, hence $\epsilon_x(\psi) \vee \epsilon_x(\chi) \in \epsilon_x(P)$, which implies $\eta = \epsilon_x(\eta) \in \epsilon_x(P)$ (since $\epsilon_x(P)$ is an ideal in $\epsilon_x(\Psi)$), hence $\eta \in P$. But this is a contradiction. Therefore we conclude that $I \cap F = \emptyset$.

Then by (DPI) in \cite{davey90} there is a prime ideal $R \in X(\Psi)$ such that $I \subseteq R$ and $R \cap F = \emptyset$. This implies $Q \subseteq R$ and $R \cap (\epsilon_x(\Psi) - P) = \emptyset$. But $R \cap (\epsilon_x(\Psi) - P) = R \cap (\epsilon_x(\Psi) - \epsilon_x(P))$, hence also $R \cap (\epsilon_x(\Psi) - \epsilon_x(P)) = R \cap \epsilon_x(\Psi) \cap \epsilon_x(P)^{c} = \emptyset$. This implies $\epsilon_x(R) = R \cap \epsilon_x(\Psi) \subseteq \epsilon_x(P)$. But we derive also from $I \subseteq R$ that $\epsilon_x(P) \subseteq R$, hence $\epsilon_x(P) = \epsilon_x(R)$.
\end{proof}

\begin{lemma} \label{leCignoli2}
If $P \in X(\Psi)$, $\phi \in \Psi$ and $\epsilon_x(\phi) \in P$, then there is a prime ideal $R \in X(\Psi)$ such that $\epsilon_x(R) = \epsilon_x(P)$ and $\phi \in R$.
\end{lemma}

\begin{proof}
If $\phi \in P$, take $R = P$. Otherwise consider the principal ideal $\downarrow\!\phi$ and the filter $F$ generated by $\epsilon_x(\Psi) - P$. Assume $\downarrow\!\phi \cap F \not= \emptyset$ and consider $\psi \in \downarrow\!\phi \cap F$. Then $\psi \leq \phi$ and there is a $\chi \in F$, $\chi \not\in P$ such that $\chi = \epsilon_x(\chi) \leq \psi$. So $\chi = \epsilon_x(\chi) \leq \psi \leq \phi$, hence $\chi = \epsilon_x(\chi) \leq \epsilon_x(\phi)$. Since $\epsilon_x(\phi) \in P$ we obtain $\chi \in P$, which is a contradiction. So $\downarrow\!\phi \cap F = \emptyset$.

Again, by (DPI) in \cite{davey90} there is a prime ideal $Q \in X(\Psi)$ so that $\downarrow\!\phi \subseteq Q$ and $F \cap Q = \emptyset$. Thus, $\phi \in Q$ and $\epsilon_x(\Psi) - P = \epsilon_x(\Psi) - \epsilon_x(P) \subseteq \Psi - Q$. But this implies $\epsilon_x(Q) \subseteq \epsilon_x(P)$. In fact, assume $\psi \in \epsilon_x(Q) \subseteq Q$, hence $\psi = \epsilon_x(\psi) \not\in \epsilon_x(\Psi) - \epsilon_x(P)$. Since $\psi \in \epsilon_x(\Psi)$, this implies $\psi \in \epsilon_x(P)$. Apply now Lemma \ref{leCignoli1} to obtain a $R \in X(\Psi)$ such that $\phi \in Q \subseteq R$ and $\epsilon_x(R) = \epsilon_x(P)$.
\end{proof}

Based on these results we are now in a position to introduce information extraction in the lattice $(2^{X(\Psi)};\subseteq)^{\vartheta}$. For any subset $S$ of $X(\Psi)$ and for any $x \in D$ define  the associated saturation operator
\begin{eqnarray} \label{eq:CylOp}
\sigma_{P_x}(S) = \{P \in X(\Psi):\exists Q \in S \textrm{ so that}\ \epsilon_x(P) = \epsilon_x(Q)\}.
\end{eqnarray}

Here comes the main theorem for the representation of distributive lattice information algebras:

\begin{theorem} \label{th:MapExtr}
For all $\phi \in \Psi$ and $\epsilon_x \in E$.
\begin{eqnarray} \label{eq:MapExtr2}
X_{\epsilon_x(\phi)} = \sigma_{\epsilon_x}(X_{\phi})
\end{eqnarray}

\end{theorem}

\begin{proof}
Consider $P \in \sigma_{\epsilon_x}(X_{\phi})$. Then there is a $Q \in X_{\phi}$ such that $\epsilon_x(Q) = \epsilon_x(P)$ and $\phi \in Q$. From $\epsilon_x(\phi) \leq \phi$ it follows that $\epsilon_x(\phi) \in \epsilon_x(Q) = \epsilon_x(P) = \epsilon_x(\Psi) \cap P$. So we see that $\epsilon_x(\phi) \in P$, hence $P \in X_{\epsilon_x(\phi)}$.

Conversely, let $P \in X_{\epsilon_x(\phi)}$, that is, $\epsilon_x(\phi) \in P$. By Lemma \ref{leCignoli2} there is a $R \in X(\Psi)$ such that $\phi \in R$ and $\epsilon_x(R) = \epsilon_x(P)$. But this means that $P \in \sigma_{\epsilon_x}(X_{\phi})$. This proves that $X_{\epsilon_x(\phi)} = \sigma_{\epsilon_x}(X_{\phi})$.
\end{proof}

The set of clopen up-sets $\mathcal{U}(X(\Psi))$ of $X(\Psi)$ is the image of $\Psi$ under the mapping $\phi \mapsto X_{\phi}$. Since the mapping is a lattice isomorphism, $\mathcal{U}(X(\Psi))$ is a distributive lattice. If $(\Psi;\mathcal{E},\cdot,1,0)$ is a distributive lattice information algebra, Theorem \ref{th:MapExtr} extends this lattice isomorphism to an information algebra isomorphism defined by the pair of maps $\psi \mapsto X_{\psi}$ and $x \mapsto P_x$. Thus $(\mathcal{U}(X(\Psi));g(\mathcal{E}),\cap,X(\Psi),\emptyset)$ becomes set algebra. We recall that the saturation operators $\sigma_{P_x}$ are existential quantifiers. We summarize this in the following representation theorem for distributive lattice information algebras:

\begin{theorem} \label{th:GenDistr.LatticeInfAlg}
A distributive lattice information algebra $(\Psi;\mathcal{E},\cdot,1,0)$ is isomorphic both as an information algebra as well as a lattice to the set algebra $(\mathcal{U}(X(\Psi));g(\mathcal{E});\cap,X(\psi),\emptyset)$.
\end{theorem}

\cite{cignoli91} extended the duality theory of distributive lattices to lattices with quantifiers. This theory could be extended to distributive lattice information algebras; we resign however to do this here, see however \cite{kohlasschmid16}. 

%
%
%
%


\section{Conditional Independence} \label{sec:CondIndep}

\subsection{Quasi-Separoids} \label{subsec:QSep}

In Section \ref{sec:SetAlgebra} a relation of conditional independence between partitions was introduced and it was shown that it has important consequences for the set algebra, see Theorem \ref{th:LocCompBas}. Conditional independence has since long be identified as an important relation for modeling and computing in probability theory, see for instance \cite{lauritzen88,pearl88,cowell99}. In the theory or relational databases too, conditional independence plays an important role \cite{maier83}, and in different other references conditional independence \cite{dawid01,dawid79,dawid98,shenoy94c,shenoy97a,studeny93,studeny95} has been studied in relation to various other formalisms of reasoning, all of which are in fact instances of information or valuation algebras \cite{kohlas03,kohlas17}, which turns out to be a basic structure to study conditional independence. Conditional independence is a fundamental issue for information. Based on the the q-separoid structure introduced above in Section \ref{sec:SetAlgebra}, this subject is discussed here in a somewhat larger context, than in the references cited above (with the exception of \cite{kohlas17}).

We recap here the theory of q-separoids from \cite{kohlas17}. We have proposed to model a system of questions or domains by a join-semilattice $(D;\leq)$, see Section \ref{sec:FomFreeAlg}. Besides order, representing granularity of questions, a further relation between questions which  describes \textit{conditional independence} of two questions, given a third one is important, as we have seen in set algebras. Therefore, in $D$ a relation $x \bot y \vert z$ is considered which is thought to express the idea that an information relative to $x$, does restrict the possible answers to $y$ only through its part relative to $z$, and vice versa. Or, in other words, only the part relative to $z$ of an information relative to $x$ is relevant as an information relative to $y$, and vice versa. Rather than to give an explicit definition of this relation in $D$ at this place, we only require it to satisfy the following four conditions:

\begin{description}
\item[C1] $x \bot y \vert y$ or all $x,y \in D$, 
\item[C2] $x \bot y \vert z$ implies $y \bot x \vert z$,
\item[C3] $x \bot y \vert z$ and $w \leq y$ imply $x \bot w \vert z$,
\item[C4] $x \bot y \vert z$ implies $x \bot y \vee z \vert z$.
\end{description}

A join-semilattice $D$ together with a relation $x \bot y \vert z$, $(D;\leq,\bot)$, satisfying conditions C1 to C4 will be called a \textit{quasi-separoid} (or also \textit{q-separoid}). In Section \ref{sec:SetAlgebra} such a relation is defined among partitions as a basic model. In the literature two additional conditions are usually added for a relation of conditional independence \cite{dawid01}:

\begin{description}
\item[C5] $x \bot y \vert z$ and $w \leq y$ imply $x \bot y \vert z \vee w$,
\item[C6] $x \bot y \vert z$ and $x \bot w \vert y \vee z$ imply $x \bot y \vee w \vert z$.
\end{description}

Then $D$ is called a \textit{separoid}. If $(D;\leq)$ is a lattice, then yet another condition can be added:

\begin{description}
\item[C7] If $z \leq y$ and $w \leq y$, then $x \bot y \vert z$ and $x \bot y \vert w$ imply $x \bot y \vert z \wedge w$.
\end{description}

With this additional condition $D$ is called a \textit{strong separoid}. For a detailed discussion of separoids we refer to \cite{dawid01}. For example it can be shown that C1 to C3 together with C5 and C6 imply C4. For our purposes C1 to C4 are sufficient. 

Families of compatible frames, as generalizations of partitions \cite{kohlas17}, provide an important example of quasi-separoids, where $D$ is in general only a join-semilattice. Here, we discuss briefly the case where $D$ is a lattice, as for example the lattice of subsets of variables. Define $x \bot_{L} y \vert z$ to hold if and only if
\begin{eqnarray} \label{eq:CondIndepLatt}
(x \vee z) \wedge (y \vee z) = z.
\end{eqnarray}

\begin{theorem}
If $(D;\leq)$ is a lattice, then the relation $x \bot_{L} y \vert z$ defines a quasi-separoid.
\end{theorem}

\begin{proof}
We have $(x \vee y) \wedge (y \vee y) = y$, hence C1 is satisfied. By the symmetry of the definition, C2 holds too. If $w \leq y$, then $z \leq (x \vee z) \wedge (w \vee z) \leq (x \vee z) \wedge (y \vee z) = z$, so C3 follows. Finally from (\ref{eq:CondIndepLatt}) we see that C4 is valid.
\end{proof}

If $x \leq y$, then from $x \bot y \vert y$ (C1) it follows that $x \bot x \vert y$ by C3. Now, in some cases $x \bot x \vert y$ implies $x \leq y$. A separoid with this property is called basic, see \cite{dawid01}. We adapt this to call a quasi-separoid  \textit{basic}, if $x \bot x \vert y$ implies $x \leq y$. The following theorem was proved in  \cite{dawid01} for a basic separoid, but it is valid for basic quasi-separoids too.

\begin{theorem} \label{th:BasicCond} 
Suppose $(D;\leq)$ is a lattice. Then a q-separoid $(D;\leq,\bot)$ is basic if and only if 
\begin{eqnarray} \label{eq:BasicSep}
x \bot y \vert z \Rightarrow (x \vee z) \wedge (y \vee z) = z.
\end{eqnarray}
\end{theorem}

\begin{proof}
If (\ref{eq:BasicSep}) holds, then $x \bot x \vert y$ implies $x \vee y = y$, hence $x \leq y$.

Suppose now that $x \bot y \vert z$. Then $(x \vee z) \bot (y \vee z) \vert z$ by C4 and C2. Define $w = (x \vee z) \wedge (y \vee z)$ such that $w \leq x \vee z$ and $w \leq y \vee z$. Using C3 and C2 we deduce then that $w \bot w \vert z$. So, if the quasi-separoid is basic, we obtain that $w \leq z$. Since we always have $w \geq z$, it follows that $w = z$.
\end{proof}

According to the first part of the proof of Theorem \ref{th:CommCondIndep} in Section \ref{sec:SetAlgebra}, the conditional independence relation in a lattice of partitions is basic.

If we meet both sides of (\ref{eq:CondIndepLatt}) with $x$ we obtain $x \wedge (y \vee z) = x \wedge z$ which is equivalent to
\begin{eqnarray} \label{eq:CondIndeModLatt}
x \wedge (y \vee z) \leq z.
\end{eqnarray}
This condition is equivalent to (\ref{eq:CondIndepLatt}) if the lattice $D$ is \textit{modular}. So, in this case we have $x \bot_{L} y \vert z$ if and only if (\ref{eq:CondIndeModLatt}) holds.

\begin{theorem} \label{th:CondIndModLatt}
If $(D;\leq)$ is a lattice, the relation $x \bot_{L} y \vert z$ defines a separoid if and only if $(D;\leq)$ is modular.
\end{theorem}

\begin{proof}
Assume $D$ modular. We are going to show that C5 and C6 are satisfied. If $D$ is modular, then $x \wedge (y \vee z) = x \wedge z$ if and only if $x \bot_{L} y \vert z$. So, if $w \leq y$, it follows $x \wedge (y \vee z \vee w) = x \wedge (y \vee z)$. Therefore, $x \wedge (z \vee w) \leq x \wedge (y \vee z \vee w) = x \wedge (y \vee z) = x \wedge z \leq x \wedge (z \vee w)$, hence $x \wedge (y \vee (z \vee w)) = x \wedge (z \vee w)$. This shows that $x \bot_{L} y \vert z \vee w$, that is C5. Further, $x \bot_{L} y \vert z$ and $x \bot_{L} w \vert y \vee z$ imply $x \wedge (y \vee z) = x \wedge z$ and $x \wedge (w \vee y \vee z) = x \wedge (y \vee z)$. Together, this leads to $x \wedge (w \vee y \vee z) = x \wedge z$, hence $x \bot_{L} (y \vee w) \vert z$. So C6 holds.

On the other hand, assume $x \bot_{L} y \vert z$ to be a separoid. By (\ref{eq:CondIndepLatt}) we have $x \bot _{L} y \vert x \wedge y$. Thus, if  $z \leq x$, by C5, it follows that $x \bot _{L} y \vert (x \wedge y) \vee z$. This means that $x \wedge (y \vee z) = (x \wedge y) \vee z$, which is modularity.
\end{proof}

Further, (\ref{eq:CondIndepLatt}) implies that 
\begin{eqnarray} \label{eq:CondIndDistLatt}
x \wedge y \leq z.
\end{eqnarray}
If the lattice $D$ is \textit{distributive}, then $(x \vee z) \wedge (y \vee z) = (x \wedge y) \vee z$. In this case (\ref{eq:CondIndepLatt}) is equivalent to (\ref{eq:CondIndDistLatt}).

\begin{theorem} \label{th:CondIndDisLatt}
If $(D;\leq)$ is a distributive lattice the relation $x \bot_{L} y \vert z$ defines a strong separoid.
\end{theorem}

\begin{proof}
A distributive lattice is modular, so C5 and C6 hold. It remains to prove C7. Assume $D$ distributive so that  $x \bot _{L}y \vert z$ if and only if (\ref{eq:CondIndDistLatt}). Now $x \bot _{L} y \vert z$ and $x \bot _{L} y \vert w$ imply $x \wedge y \leq z$ and $x \wedge y \leq w$, hence $x \wedge y \leq z \wedge w$, which shows that $x \bot _{L} y \vert z \wedge w$. Therefore C7 is satisfied.
\end{proof}

We may also consider the relation $x \bot_{d} y \vert z$ which holds if and only if $x \wedge y \leq z$. The following theorem is due to \cite{dawid01}:

\begin{theorem}
The relation $x \bot_{d} y \vert z$ is a separoid if and only if $(D;\leq)$ is a distributive lattice. 
\end{theorem}

In a distributive lattice $x \bot_{L} y \vert z$ if and only if $x \bot_{d} y \vert z$ by the discussion above. Therefore if $x \bot_{d} y \vert z$ is a separoid, it is a strong separoid by Theorem \ref{th:CondIndDisLatt}.

An important instance of a distributive lattice is the lattice of the subsets of a set $I$. If $s,t,r$ denote subsets of $I$, then $s \bot_{L} t \vert r$ if and only if $s \cap t \subseteq r$. This is then a strong separoid by the theorems above. This is the classical case of multivariate models (see Section \ref{sec:SetAlgebra}) considered in the large majority of studies on conditional independence. 

Consider two q-separoids $(D_1;\leq_1,\bot_1)$ and $(D_2;\leq_2,\bot_2)$. A map $f : D_1 \rightarrow D_2$ is called a q-separoid homomorphism, if
\begin{enumerate}
\item $f$ is a join-homomorphism, $f(x \vee_1 y) = f(x) \vee_2 f(y)$,
\item $x \bot_1 y \vert z$ implies $f(x) \bot_2 f(y) \vert f(z)$.
\end{enumerate}
If $f$ is a bijective homomorphism and $x \bot_2 y \vert z$ implies also $f^{-1}(x) \bot_1 f^{-1}(y) \vert f^{-1}(z)$, then $f$ is a q-separoid isomorphism and the two q-separoids are called isomorphic.

If $(D_1,\leq_1)$ is a join-semilattice and $(D_2;\leq_2,\bot_2)$ a q-separoid, and $f : D_1 \rightarrow D_2$ a join-homomorphism, then the relation $x \bot_1 y \vert z$ iff $f(x) \bot_2 f(y) \vert f(z)$ is a q-separoid as is easily verified. 

In the following section, the importance of conditional independence for information algebras will be shown. At this place it is important to note, that the theory of conditional independence as based on (quasi-) separoids is not only important for information algebras, where the semigroup $(\Psi;\cdot)$ is idempotent, but also more generally for non-idempotent commutative semigroups (valuation algebras, see \cite{kohlas17}). In this more general context, as for information algebras, it is important for efficient inference procedures (local computation, see \cite{laur88,shenoyshafer90,kohlasshenoy00,kohlas03} for computation in multivariate models, and \cite{kohlas17} for more general models, especially for computing with partitions). But here we do not pursue this generalization.

\subsection{Conditional Independence in Information Algebras} \label{subsec:CondIndepInfAlg}

Now we apply the theory of q-separoids to information algebras $(\Psi;\mathcal{E},\cdot,1,0)$ with strongly order-generating sets $X$. Consider the associated set algebra $(2^X;g(E),\cap,X,\emptyset)$ (Theorem \ref{th:GenRep_1}), where $g(x) = P_x$ and $P_x$ is the partition in $X$ of equivalence classes of the relation $\alpha \equiv_x \beta$ iff $\epsilon_x(\alpha) = \epsilon_x(\beta)$. Let $P_x \bot_X P_y \vert P_z$ denote the relation of conditional independence between partitions in $X$ as defined in Section \ref{sec:SetAlgebra}. Recall that $g$ is a join-homomorphism (even an isomorphism). So, lets define in $D$ a relation
\begin{eqnarray*}
x \bot_X y \vert z \textrm{ iff}\ P_x \bot_X P_y \vert P_z.
\end{eqnarray*}
Note that this relation depends on the order-generating set $X$. Acording to Section \ref{subsec:QSep} $x \bot_X y \vert z$ is then a q-separoid in the join-semilattice $(D;\leq)$. Recall that the conditional independence relation $P_x \bot_X P_y \vert P_z$ between partions $P$ in $X$ holds iff $\alpha \equiv_z \beta$ for $\alpha,\beta \in X$ implies that there is an element $\gamma \in X$ such that $\alpha \equiv_{x \vee z} \gamma$ and $\beta \equiv_{y \vee z} \gamma$. So, $x \bot_X y \vert z$ in the join semilattice $(D;\leq)$ is also defined by this condition. 

The importance of such a conditional independence relation for information algebras follows from the fact that Theorem \ref{th:LocCompBas} carries over to the information algebra.

\begin{theorem} \label{th:LocCompBasGen}
Let $(\Psi;\mathcal{E},\cdot,1,0)$ be an information algebra with strongly order-generating set $X$ and $x \bot_X y \vert z$ the relation satisfying C1 to C4 (a q-separoid) in $(D;\leq)$ defined above.
\begin{enumerate}
\item If $x \bot_X y \vert z$ and $\epsilon_x(\phi) = \phi$, $\epsilon_y(\psi) = \psi$ for $\phi,\psi \in D$, then
\begin{eqnarray*}
\epsilon_z(\phi \cdot \psi) = \epsilon_z(\phi) \cdot \epsilon_z(\psi).
\end{eqnarray*}
\item If $x \bot_X y \vert z$ and $\epsilon_x(\phi) = \phi$, then
\begin{eqnarray*}
\epsilon_y(\phi) = \epsilon_y(\epsilon_z(\phi)).
\end{eqnarray*}
\end{enumerate}
\end{theorem}

\begin{proof}
1.) By definition $x \bot_X y \vert z$ implies $P_x \bot_X P_y \vert P_z$. Further $\phi = \epsilon_x(\phi)$ implies $f(\phi) = \sigma_{g(x)}(f(\phi))$, where $f(\phi) =\ \uparrow\!\phi \cap X$ and $g(x) = P_x$ form the embedding mappings of the previous Section \ref{sec:InfAndSetAlg}. Similarly if $\psi = \epsilon_y(\psi)$ then $f(\psi) = \sigma_{g(y)}(f(\psi))$. But then we have by Theorem \ref{th:LocCompBas},
\begin{eqnarray*}
\sigma_{g(z)}(f(\phi) \cap f(\psi)) = \sigma_{g(z)}(f(\phi)) \cap  \sigma_{g(z)}(f(\psi)).
\end{eqnarray*}
If we apply the inverse maps $f^{-1}$ and $g^{-1}$ on this equation, we get $\epsilon_z(\phi \cdot \psi) = \epsilon_z(\phi) \cdot \epsilon_z(\psi)$.

2.) is proved similarily, again using Theorem \ref{th:LocCompBas}.
\end{proof}

This theorem shows that $(\Psi;\mathcal{E},\cdot,1,0)$ is a generalized information algebra in the sense of reference \cite{kohlas17}. We call the properties expressed in item 1.) and 2.) of the theorem the \textit{combination} and \textit{extraction properties} of the information algebra resepectively. They are important for computational purposes, see \cite{kohlas17}. Note that the combination property is an extension of E3. In fact, we have $x \bot_X y \vert x$ and $\epsilon_x(\phi)$ has support $x$ and hence E3 is a particular case of the combination property.

It turns out that the conditional independence relation $x \bot_Xs y \vert z$ in $(D;\leq)$ defined above is the only one possible, that is, it does not depend on $X$. This follows from the next theorem.

\begin{theorem} \label{th:UniqueCVondIndep}
Let $(\Psi;\mathcal{E},\cdot,1,0)$ be an information algebra and $x \bot y \vert z$ any relation in $(D;\leq)$ satisfying C1 to C4 and such that the combination and extraction properties relative to $x \bot y \vert z$ hold. holds. Then $x \bot y \vert z$ implies $P_x \bot_\Psi P_y \vert P_z$, where this is the conditional independence relation among partitions of the set order-generating $\Psi - \{0\}$.
\end{theorem}

\begin{proof}
Let $\phi$ and $\psi$ any two elements such that $\epsilon_z(\phi) = \epsilon_z(\psi)$. Then $\phi' = \epsilon_{x \vee z}(\phi)$ has support $x \vee z$ and $\psi' = \epsilon_{y \vee z}(\psi)$ has support $y \vee z$ and we still have $\epsilon_z(\phi') = \epsilon_z(\psi')$. Then let $\chi = \phi' \cdot \psi'$, such that

\begin{eqnarray*}
\epsilon_{x \vee z}(\chi) = \epsilon_{x \vee z}(\phi' \cdot \psi') =\phi' \cdot \epsilon_{x \vee z}(\psi').
\end{eqnarray*}
But from $x \vee z \bot y \vee z \vert z$ it follows using the extraction property that $\epsilon_{x \vee z}(\psi') = \epsilon_{x \vee z}(\epsilon_z(\psi')) =  \epsilon_{x \vee z}(\epsilon_z(\phi'))$ so that
\begin{eqnarray*}
\epsilon_{x \vee z}(\chi)  =\phi' \cdot \epsilon_{x \vee z}(\epsilon_z(\phi')) = \phi'.
\end{eqnarray*}
In the same way it follows that $\epsilon_{y \vee z}(\eta) = \psi'$. So, $\phi \equiv_z \psi$ implies that there is an element $\chi$ such that $\phi \equiv_{x \vee z} \chi$ and $\psi \equiv_{y \vee z} \chi$. But this means that $P_x \bot_\Psi P_y \vert P_z$. 
\end{proof}

Let $X$ be any strongly order-generating subset of $\Psi$. Then, as Theorem \ref{th:LocCompBasGen} shows, the combination and extraction properties relative to this relation are satisfied. So, by the last theorem, $P_x \bot_\Psi P_y \vert P_z$ holds. Since the map $g$ is injective, any q-separoid relation $x \bot_X y \vert z$ must be identical with the one induced by $P_x \bot_\Psi P_y \vert P_z$. Hence we may drop the index $X$ in the conditional independence relation.

These results extend to locally order-generating systems $X_x$ for $x \in D$. Theorem \ref{th:LocCompBasGen} still holds, since the combination and extraction property are valid in any set algebra. Then, if $x \bot y \vert z$ is defined to hold if $P_x \bot P_y \vert P_z$ holds in $U$, Theorem \ref{th:UniqueCVondIndep} applies and this conditional independence relation in $(D,\leq)$ must still be the same as $x \bot_\Psi y \vert z$.

As a consequence of these statements we have the following theorem:

\begin{theorem} \label{th:ExtrOfElinX}
Let $(\Psi;\mathcal{E},\cdot,1,0)$ be an information algebra and $X$ a strongly order-generating set for this algebra. Suppose $x \bot y \vert z$. Then, if for $\alpha,\beta \in X$ we have $\alpha \equiv_z \beta$, there is an element $\gamma \in X$ such that $\alpha \equiv_{x \vee z} \gamma$ and $\beta \equiv_{y \vee z}s \gamma$. 
\end{theorem}

This holds in particular for $X = \Psi/\{0\}$, in an atomistic information algebra it holds for atoms, for meet-irreducible elements in a finite distributive lattice information algebra and it holds for prime or maximal ideals in Boolean information algebras since these are all strongly order-generating sets in these cases.  It holds also for prime ideals in distributive lattice information algebras although these elements are not order-generating; it is in fact sufficient that an embedding exists. Whereas for general and atomistic information algebra also a direct proof of this result is available, no such proof is known so far for prime elements and prime ideals. The theorem extends also to locally order-generating sets in the following way:

\begin{theorem}
Let $(\Psi;\mathcal{E},\cdot,1,0)$ be an information algebra and $X_x$ for $x \in D$ a system of locally order-generating sets for this algebra. Suppose $x \bot y \vert z$. Then, if for $\alpha \in X_x$ and $\beta \in X_y$ we have $\epsilon_z (\alpha) = \epsilon_z(\beta)$, there is an element $\gamma \in X_{x \vee y \vee z}$ such that $\epsilon_{x \vee z}(\alpha) =  \epsilon_{x \vee z}(\gamma)$ and $\epsilon_{y \vee z}(\beta) = \epsilon_{y \vee z}(\gamma)$. 
\end{theorem}

The results of this sections clarify greatly the understanding and meaning of conditional independence in information algebras. In the next section an important special case will be addressed.

\subsection{Commutative Information and Set Algebras} \label{subsec:CommAlg}

We have noted in Section \ref{subsc:ExplInfAlg} that in many cases the extraction operators commute and the corresponding order $(D;\leq)$ forms a lattice. In this case, the set of extraction operators is closed under composition. In the general framework of information algebras $(\Psi;\mathcal{E},\cdot,1.0)$ we may start with the additional assumption that $\mathcal{E}$ is closed under composition and composition is commutative. Clearly, composition is associative, hence $(\mathcal{E},\circ)$ is an idempotent, commutative semigroup. In such a semigroup we may define an order $\epsilon_x \leq \epsilon_y$ iff $\epsilon_x \cdot \epsilon_y = \epsilon_y \cdot \epsilon_x = \epsilon_x$. It follows that $(\mathcal{E};\leq)$ is a meet-semilattice under this order

\begin{lemma}
If $(\mathcal{E},\circ)$ is an idempotent, commutative semigroup, then
\begin{eqnarray*}
\epsilon_x \cdot \epsilon_y = \epsilon_y \cdot \epsilon_x = \inf\{\epsilon_x,\epsilon_y\} = \epsilon_x \wedge \epsilon_y.
\end{eqnarray*}
\end{lemma}

\begin{proof}
By idempotency we have $\epsilon_y \circ \epsilon_x \circ \epsilon_x = \epsilon_y \circ \epsilon_x$, so we have $\epsilon_y \circ \epsilon_x \leq \epsilon_x$ and in the same way $\epsilon_y \circ \epsilon_x \leq \epsilon_y$. Consider $\epsilon_z \in \mathcal{E}$, such that $\epsilon_x,\epsilon_y \leq \epsilon_z$. Then, we have
\begin{eqnarray*}
\epsilon_z \circ \epsilon_x \circ \epsilon_y = \epsilon_z \circ \epsilon_y = \epsilon_z, 
\end{eqnarray*}
 hence $\epsilon_z \leq \epsilon_y \circ \epsilon_x$, which shows that $\epsilon_y \circ \epsilon_x$ is the infimum of $\epsilon_x$ and $\epsilon_y$ in $\mathcal{E}$.
\end{proof}

Now,  $(D;\leq)$ is always a join-semilattice. We assume that the map $x \mapsto \epsilon_x$ is injective, hence bijectieve between $(D;\leq)$ and $(\mathcal{E};\leq)$. The map is also order-preserving.

\begin{lemma}
If $x \leq y$, then $\epsilon_x \leq \epsilon_y$.
\end{lemma}

\begin{proof}
By Lemma \ref{le:SuppProp} $x$ is a support of $\epsilon_x(\psi)$, hence if $x \leq y$, then $y$ is also a support of $\epsilon_x(\psi)$. So, for any $\psi \in \Psi$ we have $\epsilon_y(\epsilon_x(\psi)) = \epsilon_x(\psi)$, hence $\epsilon_y \circ \epsilon_x = \epsilon_x$ or $\epsilon_x \leq \epsilon_y$.
\end{proof}

Therefore, the map $x \mapsto \epsilon_x$ is an order isomorphism. This implies, that the join-semilattice $(D;\leq)$ is a lattice and so is $(\mathcal{E};\leq)$. In particular, we have
\begin{eqnarray*}
\epsilon_x \cdot \epsilon_y = \epsilon_y \cdot \epsilon_x = \epsilon_{x \wedge y}.
\end{eqnarray*}

Consider now any strongly order-generating set $X$ for the information algebra $(\Psi;\mathcal{E},\cdot,1,0)$. Then we have
\begin{eqnarray*}
\uparrow\!\epsilon_x(\psi) = \sigma_{P_x}(\uparrow\!\psi)s
\end{eqnarray*}
and therefore, the saturation operators $\sigma_{P_x}$ for $x \in D$ commute, that is the set algebra $(2^X;g(\mathcal{E}),\cap,X,\emptyset)$ is a commutative set algebra. Therefore, we call an information algebra $(\Psi;\mathcal{E},\cdot,1,0)$, where $(\mathcal{E};\circ)$ is an idempotent commutative semigroup also a commutative information algebra. By Theorem \ref{th:CommCondIndep} we have then  that $P_x \bot_X P_y \vert P_z$ holds iff $(P_x \vee P_z) \wedge (P_y \vee P_z) = P_z$. This induces then in the lattice $(D;\leq)$ the unique conditional independence relation $x \bot y \vert z$, which holds iff $(x \vee z) \wedge (y \vee z) = z$. We have then in particular $x \bot y \vert x \wedge y$. So, in a commutative information algebra we have the following result, specializing Theorem \ref{th:ExtrOfElinX} 

\begin{theorem}
Let $(\Psi;\mathcal{E},\cdot,1,0)$ be a commutative information algebra and $X$ a strongly order-generating set for this algebra. Then, if for $\alpha,\beta \in X$ we have $\alpha 
\equiv_{x \wedge y} \beta$, then there is an element $\gamma \in X$ such that $\alpha \equiv_x \gamma$ and $\beta \equiv_y \gamma$. 
\end{theorem}

This holds as before, in the general case, in particular for $X = \Psi/\{0\}$, in an atomistic information algebra it holds for atoms, in a finite distributive lattice information algebra it holds for prime elements, for Boolean information algebra it holds for maximal ideals and for distributive lattice information algebras, it holds for prime ideals. The result extends also to locally order-generating sets in the following way:

\begin{theorem}
Let $(\Psi;\mathcal{E},\cdot,1,0)$ be a commutative information algebra and $X_x$ for $x \in D$ a system of locally order-generating sets for this algebra. Then, if for $\alpha \in X_x$ and $\beta \in X_y$ we have $\epsilon_{x \wedge y} (\alpha) = \epsilon_{x \wedge y}(\beta)$, then there is an element $\gamma \in X_{x \vee y}$ such that $\epsilon_x(\alpha) =  \epsilon_x(\gamma)$ and $\epsilon_y(\beta) = \epsilon_y(\gamma)$. 
\end{theorem}

This concludes our discussion of conditional independence of information, at least how this concept expresses itself in information algebras. Order, hence idempotency of combination plays an important role in our analysis. It is an open question how these results carry over to non-idempotent algebras, valuation algebras \cite{kohlas03}, where at least in regular and separative algebras \cite{kohlas17}, we also have at least a preorder, which may serve for an extension of the present results.


\section{Conclusion}

Any information algebra as defined here can be embedded into a set-algebra. Exploiting structural particularities of an information algebra there exist possibly several representations in different set algebras. This concerns especially atomistic algebras, Boolean information algebras and distributive lattice algebras. In set-algebras a conditional independence relation can be defined among partitions. Since partitions in this context represent domains or questions, this is a relation between questions or domains. Based on this conditional independence relation it can be shown that set-algebras enjoy the combination and extraction property, which is important for efficient  local computation schemes for inference in these algebras  \cite{kohlas17}. Using the set-algebra representations of information algebras this conditional independence relation can be transported to the domains underlying the information algebra and it turns out that the information algebra inherits the combination and extraction properties from the set-algebras. This shows that the weak axiomatic definition of information algebras of this paper is in fact equivalent to the one given in \cite{kohlas17}, where an abstract conditional independence relation was postulated. This result gives also a concrete interpretation of the abstract relation used in \cite{kohlas17}. 

The representation of information algebras by set-algebras is essentially based on the concept of strongly order-generating subsets of the algebra. This is only a sufficient concept for obtaining an embedding as the case of distributive lattice information algebras shows. Also it was noted that the weaker concept of locally order-generating sets are also sufficient to obtain embeddings. Locally atomistic information algebras provide an example for this. There exist more examples of such structures. We only mention locally Boolean or locally distributive lattice information algebras, structures, which we do not discuss here. 

In applications similar algebraic structures like information algebras without the idempotency axiom are also very important \cite{shenoyshafer90,kohlas03}. Such so-called valuation algebras model many different formalisms of probabilistic inference and various other uncertainty formalisms. Conditional independence plays also an important role in these formalisms, hence in valuation algebras \cite{shenoy94c,kohlas03}, and in particular q-separoids seem to be fundamental also in this context \cite{kohlas17}. Now, the results presented here depend strongly on the order between pieces of information induced by idempotency.  The question arises therefore, to what extend the results of this paper can be extended to the non-idempotent valuation algebras. As discussed in \cite{kohlas17}  order (in fact per-order) can also be defined among valuations. For the theory of separoids, pre-order is sufficient as shown in \cite{dawid01}. So, this is possibly the key to extend the present theory to the more general valuation algebras. 

\bibliography{tcslit}
\bibliographystyle{authordate3}


\end{document}